\documentclass[sigconf]{acmart}

\pdfoutput=1 

\let\ACMmaketitle=\maketitle
\renewcommand{\maketitle}{\begingroup\let\footnote=\thanks 
	\ACMmaketitle\endgroup}

\renewcommand\footnotetextcopyrightpermission[1]{}

\newif\ifcomments
\commentstrue

\usepackage{thmtools} 
\usepackage{thm-restate}
\usepackage{soul}

\usepackage{color, colortbl}

\usepackage{moresize}

\usepackage{rotating}
\usepackage{graphicx}
\usepackage{multirow}
\usepackage{geometry}
\usepackage{amsfonts}
\usepackage{amsmath}
\usepackage{mathtools}
\usepackage{wrapfig}
\usepackage{listings}
\usepackage{graphicx}
\usepackage{caption}
\usepackage{tabularx}
\usepackage{fontawesome}
\usepackage{pifont}
\usepackage{enumitem}
\usepackage{algpseudocode}
\usepackage[font={bf,sf,footnotesize}]{subfig} 
\usepackage{amsthm}
\usepackage{placeins}
\usepackage{bm}
\usepackage[framemethod=tikz]{mdframed}

\usepackage{algorithm}
\usepackage{algpseudocode}
\usepackage{algorithmicx}
\usepackage{verbatim}
\usepackage{xr}
\usepackage{float}
\usepackage{listings}

\newfloat{lstfloat}{htbp}{lop}
\floatname{lstfloat}{Listing}

\newtheorem*{corollary*}{Corollary}
\newtheorem{corollary}{Corollary}
\newtheorem{lma}{Lemma}

\newcommand\arr[1]{{#1}}
\newcommand\phii[1]{{#1}}

\newcommand\var[1]{{#1}}
\newcommand\range[1]{{#1}}
\newcommand\varvec[1]{{\textbf{#1}}}
\newcommand\vecrange[1]{{\textbf{#1}}}

\newcommand\Ao[1]{{\textbf{#1}}}
\newcommand\Bj[1]{{\textbf{#1}}}
\newcommand\Bo[1]{{\textbf{#1}}}

\newcommand\Intuition[1]{\textbf{Intuition.} \emph{#1}}
\newcommand\Note[1]{\textbf{Note.} \emph{#1}}

\definecolor{darkgreen}{rgb}{0.0, 0.5, 0.13}

\ifcomments
\newcommand\greg[1]{\textcolor{blue}{[Greg: #1]}}

\newcommand\mac[1]{\textcolor{red}{[Mac: #1]}}
\newcommand\yishai[1]{\textcolor{red}{[Yishai: #1]}}
\else
\newcommand\greg[1]{\textcolor{blue}{}}
\newcommand\toskip[1]{\textcolor{green}{}}
\newcommand\mac[1]{\textcolor{red}{}}
\newcommand\yishai[1]{\textcolor{red}{}}
\fi

\DeclareMathOperator*{\argmax}{arg\,max}
\DeclareMathOperator*{\argmin}{arg\,min}

\definecolor{darkgrey}{RGB}{70,70,70}
\definecolor{lightgrey}{RGB}{200,200,200}

\usepackage[font={bf,sf}]{caption}

\usepackage{cleveref}
\usepackage{xspace}
\usepackage[utf8]{inputenc}
\crefname{section}{§}{§§}
\Crefname{section}{§}{§§}

\newcommand*{\affmark}[1][*]{\textsuperscript{#1}}

\DeclareSymbolFont{matha}{OML}{txmi}{m}{it}
\DeclareMathSymbol{\varS}{\mathord}{matha}{83}

\colorlet{hlcolor}{yellow!20}

\mdfdefinestyle{graybox}{
	splittopskip=0,%
	splitbottomskip=0,%
	frametitleaboveskip=0,
	frametitlebelowskip=0,
	skipabove=0,%
	skipbelow=0,%
	leftmargin=0,%
	rightmargin=0,%
	innerrightmargin=0mm,%
	innerleftmargin=0mm,%
	innertopmargin=1mm,%
	innerbottommargin=0mm,%
	roundcorner=0mm,%
	backgroundcolor=hlcolor}

\newcommand{\macb}[1]{\textbf{\textsf{#1}}}

\acmConference[Technical Report]{Technical Report}{2020}{}
\startPage{1}

\setcopyright{none}

\usepackage{booktabs} 
\usepackage{makecell}

\usepackage{pifont}

\begin{document}
	
	\newcommand{\conflux}{CO$\mathit{nf}$LUX\xspace}
	\newcommand{\xparting}{\mbox{$X$-Partitioning}\xspace}
	\newcommand{\xpart}{\mbox{$X$-partition}\xspace}

\title[Near-Optimal LU Factorization]{On the Parallel I/O Optimality of Linear 
Algebra 
	Kernels:
	Near-Optimal LU Factorization}         

\author{
	Grzegorz Kwasniewski\affmark[1], 
	Tal Ben-Nun\affmark[1], 
Alexandros Nikolaos Ziogas\affmark[1],\\
Timo Schneider\affmark[1], 
Maciej Besta\affmark[1],
 Torsten Hoefler\affmark[1]\\
	{\normalsize\affmark[1]Department of Computer Science, 
	ETH Zurich\vspace{2em}}
}

\renewcommand{\shortauthors}{G. Kwasniewski et al.}
\begin{abstract}
Dense linear algebra kernels, such as linear 
solvers or tensor 
contractions, are fundamental components of many scientific computing 
applications. In this work we present a novel method of 
deriving parallel 
I/O lower bounds for this broad family of programs. 
Based on the 
\xparting 
abstraction, our method explicitly captures inter-statement dependencies.
Applying our analysis to LU factorization, we derive 
\conflux,
 an LU 
algorithm with the parallel I/O cost of $N^3 / (P \sqrt{M})$ communicated 
elements per processor --- only $1/3\times$ over our 
established lower 
bound.
We evaluate 
\conflux 
on 
various problem sizes, demonstrating empirical results that match our 
theoretical analysis, communicating asymptotically less than 
Cray ScaLAPACK or SLATE, and outperforming the asymptotically-optimal 
CANDMC library. Running on $1$,$024$ nodes of Piz Daint, 
\conflux 
communicates 1.6$\times$ less than the second-best 
implementation and is 
expected to communicate 2.1$\times$ less on a full-scale run 
on Summit.
\end{abstract}

\maketitle

\section{Introduction}
\label{sec:intro}

Data movement is widely considered a bottleneck in high- performance
computing~\cite{survey}, often dominating time and energy consumption of
computations~\cite{kestor2013quantifying,padal}. Thus, deriving 
algorithmic I/O lower bounds has always been of 
theoretical interest~\cite{general_arrays, redblue}; and
developing I/O-efficient schedules is of high
practical value~\cite{maciejBC, edgarTradeoff}.
In linear 
algebra computations,
this challenge is exacerbated by the fact that the matrices 
of interest can be
prohibitively large. 
Simultaneously, large-scale linear algebra kernels such 
as matrix 
factorizations~\cite{meyer2000matrix,krishnamoorthy2013matrix}
 or tensor
contractions~\cite{solomonik2014massively}, are the basis of 
many problems in
scientific computing~\cite{joost, rectangularML}. Therefore,
accelerating these routines is of great 
significance for numerous domains.

Analyzing I/O bounds of linear algebra kernels dates back to 
a seminal work 
by 
Hong 
and Kung{~\cite{redblue}}, who derived a first asymptotic 
bound for 
matrix-matrix 
multiplication (MMM) using the red-blue pebble game 
abstraction. This method 
was 
subsequently extended and used by other works to derive 
asymptotic{~\cite{ElangoSymbolic}} and tight{~\cite{COSMA}} 
bounds for more 
complex programs.
Despite its expressibility, problems based on pebble game 
abstractions are
notoriously hard to solve, as they are P-SPACE complete in 
the general 
case{~\cite{redblueHard_}}.
Other techniques include methods based on the 
Loomis-Whitney 
inequality{~\cite{IronyMMM}}, {\cite{cholesky1}}, 
{\cite{ballard2011minimizing}},
{\cite{anotherLU}} and the polyhedral model program 
representation{~\cite{benabderrahmane2010polyhedral}}.
Ultimately, the
existing methods are either problem-specific and hard to 
generalize{~\cite{COSMA}}; provide only asymptotic or 
non-tight lower 
bounds
{~\cite{demmel4}},{~\cite{IronyMMM}}; or are limited to only 
single-statement 
micro kernels, unable to capture more complex 
dependencies{~\cite{general_arrays}},{~\cite{comm_DNN}}.

To tackle these challenges, we first provide a 
\emph{general} 
method
for deriving \emph{precise} I/O lower bounds of Disjoint 
Array Access 
Programs 
(DAAP) --- a broad range
of programs composed of a sequence
of statements enclosed in an arbitrary number of nested loops.
Within this class, we explicitly model both the 
\textit{per-statement data 
	dependencies}, using the 
{\xparting} 
abstraction{~\cite{COSMA}}, as well as 
\textit{inter-statement data 
dependencies}, in which we model potential data reuse.
In Section~\ref{sec:lu_lowerbound} we illustrate the applicability of our 
framework to derive a
parallel I/O lower bound of LU 
factorization: $\frac{2}{3}\frac{N^3}{P \sqrt{M}}$ elements, 
where $N$ is the  
matrix size, $P$ is the number of processors, and $M$ is the 
local memory size.

Moreover, in Section~\ref{sec:conflux}, we use the 
insights from deriving
the above lower bound to develop \conflux, a near \emph{Communication
	Optimal LU factorization, \xparting{}-based} algorithm.  
Our
algorithm minimizes data movement across the 2.5D processor decomposition 
using 
a 
row-masking 
tournament pivoting strategy, resulting in a communication requirement of  
$\frac{N^3}{P \sqrt{M}} + \mathcal{O}\big(\frac{N^2}{P}\big)$ 
elements per
processor, which leading order term is only a factor of 
$\frac{1}{3}$ 
over the lower bound.

In Section~\ref{sec:evaluation}, we measure the communication 
volume of
\conflux
and we compare to other modern implementations of LU
factorization. We consider a vendor- optimized ScaLAPACK from Cray's 
LibSci~\cite{scalapack} (an implementation tuned for Cray supercomputers based
on 2D decomposition), CANDMC~\cite{candmc, candmccode} (code based on
asymptotically optimal 2.5D decomposition), and SLATE~\cite{slate} 
(a recent library targeting exascale systems with an LU implementation based 
on 
2D decomposition). 
As the scope of this
work is the I/O complexity, 
we focus on the communication volume of these implementations.
We tested them on a wide range of problem 
sizes and numbers of
processors inspired by real scientific applications. 
In our experiments on Piz Daint, we measure up to 4.1x communication reduction 
compared to the second-best implementation. Furthermore, our 2.5D 
decomposition is asymptotically better than SLATE and LibSci, with even 
greater expected speedups on exascale machines. Compared to the 
communication-avoiding 
CANDMC library with the I/O cost of $5N^3/(P\sqrt{M})$ elements~\cite{2.5DLU}, 
\conflux communicates five times less.

\noindent
In this work, we provide the following contributions:
\begin{itemize}[leftmargin=*]
	\item A general method
	for deriving parallel I/O lower bounds of a broad
	range of linear algebra kernels.
	\item An I/O lower bound of parallel LU factorization.
	\item \conflux, a provably near-I/O-optimal parallel 
	algorithm for LU factorization.
	\item A full analysis of communication volume in \conflux and a
	comparison to the state-of-the-art implementations of LU factorization
	(LibSci, SLATE, CANDMC), showing consistent benefits
	of \conflux and thus our general 
	approach over state-of-the-art libraries.
\end{itemize}

\section{Background}
\label{sec:background}
\subsection{Machine Model}
\label{sec:machineModel}

To model the algorithmic I/O complexity, we start with a model of a 
sequential machine equipped with a two-level deep memory hierarchy 
(Sections~\ref{sec:boundsSingleStatement} and 
~\ref{sec:mult_statements}).
In Section~\ref{sec:parredblue}, we use  
the parallel machine model and show which complexity properties 
are invariant.

\noindent \macb{Sequential machine}. A computation is performed on a 
sequential machine 
with a fast memory of limited size and unlimited slow memory. The 
fast memory can hold up to $M$ elements at any given time. 
To perform any computation, all input elements must reside in fast 
memory, and the result is stored in fast memory.

\noindent \macb{Parallel machine}. The sequential model is extended to a 
machine equipped with $P$ processors, each equipped with a private 
fast memory of size $M$. There is no global memory of unlimited 
size --- instead, elements are transferred between processors' fast 
memories.

\subsection{Input Programs}
\label{sec:inputPrograms}

We consider a general class of programs that operate on 
multidimensional
arrays. Array
elements can be loaded from slow to fast memory, stored from fast to 
slow memory, and computed inside fast memory.  
Elements have \emph{versions}, which 
are incremented every time they are 
updated. 
We model the program execution as a computational directed acyclic 
graph (cDAG, details in Section~\ref{sec:pebblegame}), where each vertex 
corresponds to a different version of an element.
E.g., for a statement $A[i,j] 
\leftarrow f(A[i,j])$, a vertex corresponding to $A[i,j]$ \emph{after} 
applying $f$ is different from a vertex corresponding to $A[i,j]$ 
\emph{before} 
applying $f$.
In a cDAG, we model it as an edge from vertex $A[i,j]$ before $f$ to 
vertex $A[i,j]$ after $f$. Initial versions of each element do not have 
any incoming edges and thus form the cDAG inputs. 
\emph{The distinction between elements and vertices} is important for our I/O 
lower bounds analysis, as we will investigate how many vertices
are computed
for a 
given number of loaded vertices.

A program is a sequence of statements $S$ enclosed in loop 
nests, each of the following form (we use the loop nest notation used 
by Dinh and 
Demmel{~\cite{anotherDemmel}}):

{\small
	\begin{align}
	\nonumber
	\text{\textbf{for }} r^1 \in R^1,
	\text{\textbf{for }}  r^2 \in R^2(r^1),
	\dots
	\text{\textbf{for }}  r^l \in R^l(r^1, \dots, r^{l-1}): \\
	\nonumber
	S: A_0[\bm{\phi_0}(\bm{r})] \leftarrow 
	f(A_1[\bm{\phi_1}(\bm{r})], A_2[\bm{\phi_2}(\bm{r})], \dots, 
	A_m[\bm{\phi_m}(\bm{r})])
	\end{align}
}

\begin{table}[t]
	\setlength{\tabcolsep}{2pt}
	\renewcommand{\arraystretch}{1.2}
	\centering
	\footnotesize
	\sf
		\begin{tabular}{@{}l|ll@{}}
			\toprule
			\multirow{7}{*}{\begin{turn}{90}\textbf{Input prog. 
						(\S~\ref{sec:inputPrograms})}\end{turn}}
			& $A_0$ & Output of statement $S$. \\ 
			& $A_j,$~~~$j = 1,\dots, m$ & Input $j$ of statement $S$. \\ %
			& $\bm{r} = \left[r^1, \dots, r^l\right]$ 
			& 
			Iteration vector composed of $l$ iteration variables.\\
			& $R^t$
			& 
			\makecell[l]{Iteration domain of variable $r^t \in R^t$, which 
			may\\ 
				depend on iteration 
				variables $1\dots t-1$.\vspace{0.2em}}\\ 
			& $\bm{\phi}_j$ &  \makecell[l]{Access vector mapping 
				$dim(\bm{\phi}_j)$ iteration 
				variab-\\les to a $dim(A_j)$ dimensional address in array 
				$A_j$.}\\
			\midrule
			\multirow{10}{*}{\begin{turn}{90}\textbf{ 
						\xparting (\S~\ref{sec:pebblegame})}\end{turn}}
			& $G = (V,E)$& \makecell[l]{\textbf{c}omputational 
			\textbf{D}irected 
				\textbf{A}cyclic \textbf{G}raph (cDAG) with\\ $V$ vertices and 
				$E 
				\subset V \times V$ directed edges.}\\
			& $M$ & Number of red pebbles (size of the fast memory).\\
			& $V_h \subset V$ & \makecell[l]{An $h$-th subcomputation of an \\ 
				$X$-partition, 
				$h = 1,\dots, s$} \\
			& $\mathit{Dom}(V_h)$ & Dominator set of 
			subcomputation $V_h$.\\
			& $\mathcal{P}(X) = \left\{V_1,\dots,V_s\right\}$ &  
			\makecell[l]{An 
				\xpart composed of $s$ disjoint\\subcomputations.}\\
			& $\Pi(X)$ & The set of all \mbox{$X$}-partitions of size $X$.\\
			& $Q$ & A number of I/O operations of a schedule. \\
			& $\rho_h$ & The computational intensity of subcomputation $V_h$.\\
			& $\rho = \max_h\{\rho_1, \dots, \rho_s\}$ & The maximum 
			computational intensity 
			of $\mathcal{P}(X)$.\\
			\midrule
			\multirow{1}{*}{\begin{turn}{90}
					\textbf{\hspace{1.5em}DAAP sched. 
						(\S~\ref{sec:boundsSingleStatement})}
			\end{turn}} 
			& $R^t_h$ &  \makecell[l]{Set of all values iteration variable $t$ 
				takes\\ during 
				subcomputation $h$.	\vspace{0.2em}}\\
			& $R^k_{h,j}$ &  \makecell[l]{Set of all values $k$-th iteration 
				variable of access \\ 
				function vector $\phi_j$ takes during subcomputation 
				$h$.\vspace{0.2em}}\\
			& $\bm{R}_h$ &  \makecell[l]{Iteration domain of subcomputation $h$ 
				--- 
				set of all \\
				iteration vectors accessed during $h$.\vspace{0.2em}} \\
			& $|A_j(\bm{R}_h)|$ &  \makecell[l]{Number of different vertices 
				accessed \\from array 
				$A_j$ during subcomputation $h$.}\\		
			\bottomrule
		\end{tabular}%
	\caption{
		\textmd{Notation used in the paper.}
	}
	\label{tab:symbols}
\end{table}

\noindent
where (cf.~Figure~\ref{fig:prog_rep} and Table~\ref{tab:symbols} for
summaries):

\begin{enumerate}[leftmargin=1.0em]
	\item The \emph{statement} $S$ is nested in a loop nest of depth $l$.
	
	\item Each loop in the $t$-th level, $t = 1,\dots, l$ 
	is associated with 
	its \emph{iteration variable} \var{$r^t$}, which 
	iterates over 
	its set \linebreak \var{$r^t$} $\in$ \range{$R^t$}. Set \range{$R^t$} may 
	depend on iteration 
	variables from outer loops \var{$r^1$}$, \dots,$ 
	\var{$r^{t-1}$} (denoted as 
	\range{$R^t$}(\var{$r^1$}$, \dots, $\var{$r^{t-1}$})).
	
	\item All $l$ iteration variables form the \emph{iteration vector} 
	\linebreak
	\varvec{$\mathbf{r}$} $= [r^1, \dots, r^l]$ and we 
	define the \emph{iteration 
		domain} \vecrange{$\bm{R}$} as the set of 
	all iteration 
	vectors $\forall \bm{r} : \bm{r} \in \bm{R}$.
	
	\item \label{enum:arr_dim} {Each evaluation of statement $S$ is a 
		function on $m$ input elements, 
		each input belongs to a logical array {\arr{$A_j$}}. Different logical 
		arrays 
		may 
		refer to the same memory region. The 
		dimension of a logical array is denoted as $dim(A_j)$.}
	
	\item Elements of logical array \arr{$A_j$} are 
	referenced 
	by an \emph{access 
		function vector} \phii{$\bm{\phi_j}$} $= [\phi_j^1, \dots, 
	\phi_j^{dim(A_j)}]$, 
	which maps $dim(A_j)$ iteration variables to a \emph{unique} element in 
	array \arr{$A_j$} (access function vector is injective). Only 
	vertices associated with the newest element versions can be 
	referenced.
	
	\item \label{enum:ref_versons} 
	A given vertex can be 
	referenced by only one access 
	function vector per statement. We will refer to this as \emph{disjoint 
		access property}.
	
	\item {The 
		\emph{access dimension} of \mbox{$A_j(\bm{\phi}_j)$}, denoted 
		\mbox{$dim(A_j(\bm{\phi}_j))$}, is the number of 
		different 
		iteration variables present in \mbox{$\bm{\phi}_j$}.
		\emph{Example: consider access \mbox{$A_j[k,k]$} used, e.g., in 
			LU factorization. Its 
			access function vector 
			\mbox{$\bm{\phi}_j$} = \mbox{$[k,k]$} is a function of only one 
			iteration variable \mbox{$k$}. Therefore, \mbox{$dim(A_j) = 2$}, 
			but 
			\mbox{$dim(A_j(\bm{\phi_j})) = 1$}.}
	}
	If it is clear from the context, we will refer to 
		\mbox{$dim(A_j(\bm{\phi}_j))$} simply as \mbox{$dim(\bm{\phi}_j)$}.
	
	\item The result of a statement evaluation is stored in array 
	\arr{$A_0$}.
\end{enumerate}

\begin{figure*}[t]
	\includegraphics[width=2.1\columnwidth]{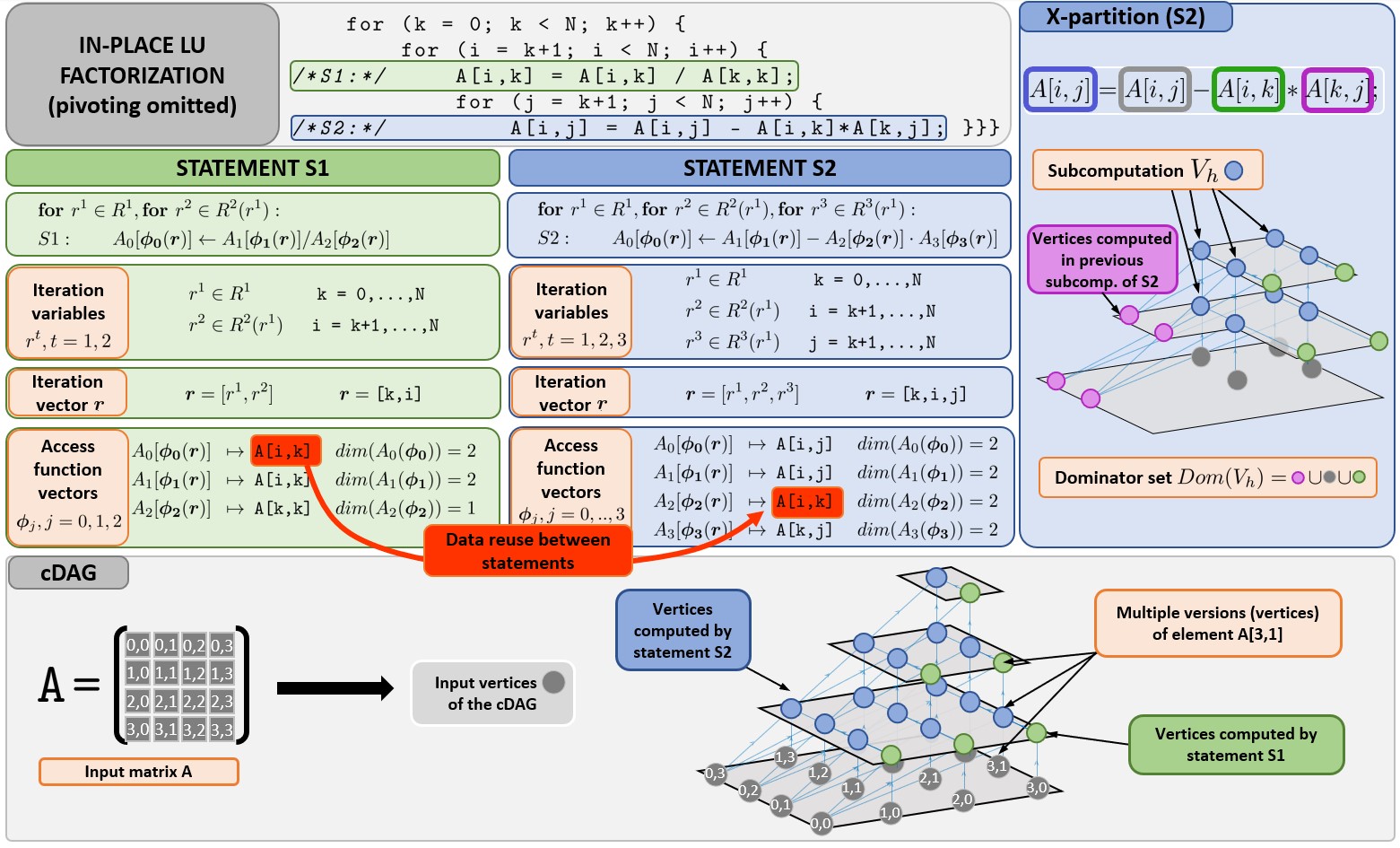}
	\caption{In-place LU factorization (for simplicity, no 
	pivoting is 
		performed). LU contains two statements ($S1$ and 
		$S2$), for which we 
		provide key components of our program representation, 
		together with the 
		corresponding cDAG for $N=4$. For 
		statement $S2$, we also provide a graphical 
		visualization of a single 
		subcomputation $V_h$ in its \xpart.}
	\label{fig:prog_rep}
\end{figure*}

We denote an input program of this form as a \textit{Disjoint Array Access 
Program} 
	(DAAP).
In summary, for each innermost loop iteration (and its corresponding 
iteration 
vector $\mathbf{r}$), each statement is an evaluation of some function $f$ 
on 
$m$ 
inputs, 
where every input is an element of array $A_j, j = 1,\dots, 
m$, 
and the result of $f$ is stored to the output array $A_0$ at location 
$\boldsymbol{\phi}_0(\mathbf{r})$. The notation used in this work {is 
summarized in Table{~\ref{tab:symbols}}}, along  
with an example program (LU factorization) in 
Figure~\ref{fig:prog_rep}.
We want to emphasize that even though the evaluation in this paper focuses 
mostly on the I/O minimization of the parallel LU factorization for 
illustrative 
purposes, our universal method can be applied to other kernels, 
like Cholesky
and QR factorizations, or more general tensor contractions.

\vspace{0.5em}
\noindent
\textit{\textbf{Note: Elements and vertices.}
	Consider a program:}
\begin{lstlisting}[numbers=none]
for k = 1:10    for i = k+1:10    for j = k+1:10
A(i,j) = A(i,j) - A(i,k)*A(k,j)
end; end; end;
\end{lstlisting}
\emph{
	Consider the element} \texttt{A(5,3)}. \emph{Even though it is referenced 
	more 
	than once, for example for}
\texttt{k=1;i=5;j=3;} \emph{by access $A_1(\bm{\phi}_1)$}=\texttt{A(i,j)}, 
\emph{and for }\texttt{k=3;i=5;j=4;} \emph{(access
	$A_2(\bm{\phi}_2)$}= \texttt{A(i,k))}, \emph{this element has 
	been updated 
	and 
	has 
	\emph{different} {versions} in these two accesses, corresponding to 
	\emph{different vertices} in the cDAG.}
\emph{Observe however, that if the second loop iterated over 
		range }\texttt{for i 
		= 
		k:10},
\emph{ 
		this would not be a valid DAAP program, as it would invalidate the 
		disjoint 
		access property.
	}

\subsection{I/O Complexity and Pebble Games}
\label{sec:pebblegame}
{We now establish the relationship between DAAP and the red-blue pebble game 
	--- a powerful abstraction for deriving lower bounds and optimal schedules 
	of 
	cDAGs evaluation.}

\subsubsection{cDAG and Red-Blue Pebble Game}

Introduced by Hong and Kung~\cite{redblue}, the red-blue pebble game is played 
on the 
computation directed acyclic graph (cDAG) $G=(V,E)$.
Every vertex $v \in V$ represents a result of a unique 
computation stored in some memory and a 
directed edge  $(u,v) \in E$ represents a data dependency.
Vertices without any incoming (outgoing) edges are called \emph{inputs} 
(\emph{outputs}). 
The vertices that are currently in fast memory are marked by a red pebble 
on  the corresponding vertex of the cDAG. Since the size of fast 
memory is 
limited (we denote this size by the parameter $M$), we can never have more than 
$M$ 
red 
pebbles on the cDAG at any moment.
Analogously, the contents of the slow memory (of unlimited size) is 
represented by an unlimited number of blue pebbles.
To perform a computation, i.e., to evaluate the value 
corresponding to vertex $v$, all direct 
predecessors of 
$v$ must be loaded into fast memory.

\noindent \macb{Rules and goal of the game.} The game proceeds as follows: 
First, 
all input vertices have blue pebbles placed on them, and no red 
pebbles are present in the cDAG. 
 At any time, one of 
the 
following 
\emph{pebbling moves} are allowed: 1) placing a red pebble on a vertex which 
has a blue pebble (load), 2) placing a blue pebble on a vertex 
which 
has a red pebble (store), 3) placing a red pebble on a vertex 
which 
all direct predecessors have red pebbles (compute), 4) removing any pebble 
from a vertex (discard). The goal of a game is to find a sequence of 
pebbling 
moves such that all output vertices have blue pebbles placed on them, and 
the 
number of load and store operations is minimized.
For this, we need definitions of certain sets of vertices that impose
a structure on the cDAG.

\subsubsection{Dominator and Minimum Sets}
\label{sec:redblue} 
For any subset of vertices $V_h 
\subset V$, a \emph{dominator set} 
$\mathit{Dom}(V_h)$ is a set such 
that every path in the cDAG from an input vertex that enters 
$V_h$ must 
contain at least one vertex in 
$\mathit{Dom}(V_h)$. 
They further define the \emph{minimum set} 
$\mathit{Min}(V_h)$ as the 
set of all vertices in $V_h$ that do not have any 
immediate successors in $V_h$.
To avoid the ambiguity of non-uniqueness of dominator set 
size,
we denote a \emph{minimum dominator set} 
$\mathit{Dom}_{min}(V_h)$ to be a dominator set with the 
smallest size.

\noindent
\Intuition{A dominator set abstracts a 
	set of 
	inputs 
	required to execute subcomputation $V_h$ and
	a minimum set a set of outputs 
	of $V_h$. We bound
	computation ``volume''
	(number of vertices in $V_h$) by its communication ``surface'', comprised 
	by 
	its inputs - vertices in
	$\mathit{Dom}_{min}\left(V_h\right)$ 
	and outputs - vertices in 
	$\mathit{Min}(V_h)$.}

\subsubsection{\xparting}
\label{sec:xpart}
Introduced by Kwasniewski et al.~\cite{COSMA}, \linebreak 
\xparting generalizes the 
work by Hong and 
Kung~\cite{redblue}.
An \xpart of a cDAG 
is a 
collection of $s$ mutually disjoint subsets $V_h$ (referred as 
\emph{subcomputations})
$\mathcal{P}(X) = \{V_1, 
\dots, V_s\}$ of $V$ to $s$  with two additional 
properties:
\begin{itemize}[leftmargin=*]
	\item no cyclic dependencies between 
	subcomputations, 
	\item $\forall h$, $\left|{Dom}_{min}\left( 
	V_h \right)\right| \le X$ and  $\left| 
	{Min}\left(V_h\right) \right| \le 
	X$.
\end{itemize}

{For a given cDAG 
	and for any given \mbox{$X > M$}, denote 
	\mbox{$\Pi(X)$} a set of all its valid \mbox{$X$}-partitions, 
	\mbox{$\mathcal{P}(X) \in \Pi(X)$}. }	
Kwasniewski et al. prove that an I/O optimal schedule of $G$, which 
performs $Q$ load and 
store operations, has an associated \xpart $\mathcal{P}_{opt}(X) \in \Pi(X)$ 
with size 
$|\mathcal{P}_{opt}(X)| \le \frac{Q + X - M}{X - M}$ {for any \mbox{$X > 
		M$}} (\cite{COSMA} extended 
version, 
Lemma 2).

\noindent
\Intuition{If a smallest dominator set of 
\mbox{$V_h$} 
contains \mbox{$X$} 
vertices, then at least \mbox{$X-M$} vertices need to loaded. 
 Note that there may exist a valid $X$-partition 
\mbox{$\mathcal{P}_{min}(X) \in \Pi(X)$} such that 
\mbox{$|\mathcal{P}_{min}(X)| < |\mathcal{P}_{opt}(X)|$}. Such 
$X$-partition 
cannot be directly translated to a valid schedule, but may serve as a lower 
bound.
	}

\subsubsection{Deriving lower bounds}

The following lemma bounds the number I/O operations required to pebble a given 
cDAG:

\begin{lma}
	\label{lma:compIntensity}
	(Lemma 4 in \cite{COSMA}, extended version)
	For any constant $X_c$, the number of I/O operations $Q$ required to pebble 
	a 
	cDAG $G=(V,E)$ with $|V| = n$ vertices using $M$ red pebbles is bounded 
	by $Q \ge {n}/{\rho}$, where ${\rho}  = \frac{|V_{max}|}{X - M}$ is the 
	maximal 
	computational intensity, $V_{max} = \argmax_{V_h \in 
		\mathcal{P}(X_c)} |V_h|$ is the largest 
	subcomputation 
	among all valid \mbox{$X_c$-partitions}.
\end{lma}

\noindent
\macb{Limitations of existing methods. }{
	While pebbling-based approaches have been successfully 
	applied 
	to 
	algorithms like FFT{~\cite{redblue}}, sorting{~\cite{redbluewhite}}, and 
	parallel MMM{~\cite{COSMA}}, they still pose several limitations: 
}
\begin{itemize}[leftmargin=*]
	\item {\textbf{Parametric cDAGs.} 
		Existing 
		methods operate on cDAGs where vertices and edges are 
		explicitly provided. To handle cDAGs of parametric sizes (e.g., 
		\mbox{$N^3$} 
		vertices of MMM or \mbox{$N \log N$} vertices of FFT), additional, 
		non-generalizable methods must be further applied.}
	\item {\textbf{Complexity.} Finding an optimal pebbling sequence is 
	P-SPACE 
		complete{~\cite{redblueHard_}}; and \mbox{$S$}- or, more general, 
		\mbox{$X$}-partitioning 
		is NP-hard (reducible to max-cut). }
	\item {\textbf{Lower bounds vs. schedule.} There 
		is no general, direct method to translate lower bounds derived from 
		\mbox{$X$}-partitioning to a correct schedule. }
\end{itemize}

{
	In the following section, we take advantage of a DAAP structure 
	(Section{~\ref{sec:inputPrograms}}) to build up a new, general method for
	obtaining I/O lower bounds.
	This allows capturing \emph{parametric cDAGs}, as all vertex sets are 
	symbolic. It drastically reduces the \emph{complexity}, as individual 
	vertices do not need to be modeled anymore. 
}

\section{General I/O Lower Bounds}
\label{sec:boundsSingleStatement}
In this section, we derive I/O 
bounds for a single statement.
In Section~\ref{sec:mult_statements} we extend our 
analysis to capture interactions and reuse between 
different statements in the program.

In this paper, we present the key lemmas and the intuition 
		behind them to guide the reader to our main result --- near optimal 
		parallel LU 
		factorization. However, the method covers a much wider spectrum of 
		algorithms. 
		For curious readers, we present all proofs of provided lemmas, together 
		with the full theoretical analysis, in the
		attached supplementary material.

We start with stating our key lemma:

\begin{restatable}{lma}{compIntensityPhi}	
	\label{lma:compIntensityPhi}
	If $|V_{max}|$ can be expressed as a closed-form function of $X$, that 
	is 
	$|V_{max}| = 
	\psi(X)$, 
	then the lower bound on $Q$ may be expressed as:
	$$Q \ge n \frac{(X_0 - M)}{\psi(X_0)}$$
	where $X_0 = \argmin_X \rho = \argmin_X \frac{\psi(X)}{X-M}$.
\end{restatable}

\begin{proof}
	Note that Lemma~\ref{lma:compIntensity} is valid for any $X_c$
	(i.e.,  
	for any  
	$X_c$, it gives a valid lower bound). Yet, these bounds are not necessarily 
	tight. 
	As we want to 
	find tight I/O lower bounds, we need to maximize the lower 
	bound. $X_0$ by 
	definition minimizes $\rho$; thus, it maximizes the bound.
	Lemma~\ref{lma:compIntensityPhi} then follows directly from 
	Lemma~\ref{lma:compIntensity} 
	by 
	substituting $\rho = \frac{\psi(X_0)}{X_0 - M}$.
\end{proof}

\noindent
\Intuition{$\psi(X)$ expresses computation ``volume'', while $X$ is its input 
	``surface''. $X_0$ corresponds to the 
	situation 
	where the ratio of this ``volume'' to the required communication is 
	minimized 
	(corresponding to a highest lower bound).}

\Note{}
If function $\psi(X)$ is differentiable and has a global 
minimum, we 
can 
find 
$X_0$ by, e.g., solving the equation $\frac{d\frac{\psi(X)}{X 
		- M}}{dX} = 0$.
The key limitation is that it is not always possible to find $\psi$, that 
is, 
to 
express $|V_{max}|$ solely as a function of $X$. However, for many linear 
algebra kernels $\psi(X)$ exists. Furthermore, one can relax this problem 
preserving the correctness of the lower bound, that is, by finding a 
function
$\hat{\psi}: \forall_X \hat{\psi}(X) \ge \psi(X)$.

\subsection{Iteration vector, domain, and access sizes}
\label{sec:access_sizes}

Each execution of statement $S$ is associated with the 
\emph{iteration vector} 
\varvec{r} = [\var{$r^1$}$, \dots, $\var{$r^l$}] $\in \mathbb{N}^l$ 
representing 
the current 
iteration, 
that is, values of iteration variables \var{$r^1$}$, \dots, $\var{$r^l$}. 
Each 
subcomputation $V_h$ is uniquely defined by all iteration 
vectors associated 
with vertices pebbled 
in $V_h$:
$\{$\varvec{$\mathbf{r}_h^1$}$, \dots, 
$\varvec{$\mathbf{r}_h^{|V_h|}$}$\} = 
\mathbf{R}_h$.
For each iteration variable \var{$r^t$}, $t = 1, \dots, l$, denote the 
set 
of all values that 
\var{$r^t$} 
takes during $V_h$ as $R^t_h$.
We have \var{$r_h^t$} $\in R^t_h \subseteq$ \range{$R^t$} 
$\subset \mathbb{N}$.
We denote $\mathbf{R}_h \subseteq [R^1_h, \dots, R^t_h] 
\subseteq $ 
\vecrange{R} as 
the \emph{iteration 
	domain} of subcomputation $V_h$.

Furthermore, recall that each input access $A_j[\bm{\phi_j}(\bm{r})]$ is 
uniquely defined by $dim(\bm{\phi}_j)$ iteration variables $r_j^1, \dots, 
r_j^{dim(\bm{\phi}_j)}$. Denote the set of all values each of $r_j^k$ takes 
during 
$V_h$ as $R^k_{h,j}$.
Given $\mathbf{R}_h$, we also denote the number of 
different vertices that are 
accessed from each input array \arr{$A_j$} as 
$|A_j(\mathbf{R}_h)|$.

We now state the lemma which bounds $|V_h|$ by the iteration sets' sizes 
$|R^t_{h}|$:

\begin{restatable}{lma}{rectTiling}
	\label{lma:rectTiling}
	{Given the ranges of all iteration variables \mbox{$R_h^t,$} \mbox{$t = 
			1,\dots,l$} 
		during 
		subcomputation 
		\mbox{$V_h$}, if \mbox{$|V_h| = \prod_{t=1}^{l}|R^t_{h}|$}, then 
		\linebreak
		\mbox{$\forall j = 1, \dots, 
			m :|A_j(\mathbf{R}_h)| = $} \mbox{$
			\prod_{k=1}^{dim(\bm{\phi}_j)}|R^k_{h,j}|$} and \mbox{$|V_h|$} is 
			maximized 
		among 
		all valid subcomputations which iterate over \mbox{$\bm{R}_h = [R_h^1, 
			\dots, R_h^t]$}.}
\end{restatable}

To prove it, we now introduce two auxiliary lemmas:
\begin{lma}
	\label{lma:vi_bound}
	For statement $S$, the size $|V_h|$ of 
	subcomputation $V_h$ (number 
	of 
	vertices of $S$ computed during 
	$V_h$) is bounded by the sizes of the iteration 
	variables' 
	sets $R_h^t, t = 
	1, \dots, l$:
	\begin{equation}
	\label{eq:vmax_vol}
	|V_h| \le \prod_{t=1}^{l}|R^t_h|.
	\end{equation}
\end{lma}

\begin{proof}
	Inequality~\ref{eq:vmax_vol} follows from a combinatorial argument: 
	each 
	computation in $V_h$ is uniquely defined by its iteration vector $[r^1, 
	\dots, 
	r^l]$. As each iteration variable $r^t$ 
	takes $|R_h^t|$ 
	different values during $V_h$, we have $|R_h^1| \cdot |R_h^2| 
	\cdot 
	\dots \cdot |R_h^t| 
	= \prod_{t=1}^{l}|R^t_{h}|$ ways how to uniquely choose the iteration 
	vector in 
	$V_h$. 
\end{proof}

Now, given $\mathbf{R}_h$, we want to assess how many 
different vertices are 
accessed for each input array \arr{$A_j$}.
Recall that this number is denoted as access size 
$|A_j(\mathbf{R}_h)|$.

We will apply the same combinatorial reasoning to 
$A_j(\mathbf{R}_h)$. For each access 
$A_j[\bm{\phi}_j(\bm{r})]$, each 
one of $r_j^k$, $k 
= 1, 
\dots, dim(\bm{\phi}_j)$ iteration variables loops over set 
$R_{h,j}^k$ during subcomputation $V_h$.
We can thus bound size of $A_j(\mathbf{R}_h)$ 
similarly to 
Lemma~\ref{lma:vi_bound}:

\begin{lma}
	\label{lma:projection_bound}
	The access size  $|A_j(\mathbf{R}_h)|$ of subcomputation 
	$V_h$ 
	(the number of vertices 
	from the array $A_j$ required to compute $V_h$) is 
	bounded by the 
	sizes of $dim(\bm{\phi}_j)$
	iteration variables' sets $R_{h,j}^k, k = 
	1, \dots, dim(\bm{\phi}_j)$:
	
	\begin{equation}
	\label{eq:a_j_vol}
	\forall_{j = 1, \dots, m} : |A_j(\mathbf{R}_h)| \le 
	\prod_{k=1}^{dim(\bm{\phi}_j)}|R^k_{h,j}|
	\end{equation}
	
	where \range{$R^k_{h,j}$} $\ni$ \var{$r^k_j$} is the set over which 
	iteration 
	variable \var{$r^k_j$}
	iterates 
	during $V_h$.
\end{lma}

\begin{proof}
	We use the same 
	combinatorial argument as in Lemma~\ref{lma:vi_bound}. Each vertex in 
	$A_j(\mathbf{R}_h)$ is uniquely defined by $[r_j^1, 
	\dots, 
	r_j^{dim(\bm{\phi}_j)}]$. Knowing the number of different values each 
	$r_j^k$ 
	takes, 
	we bound the number of different access vectors 
	$\bm{\phi}_j(\bm{r}_h)$.
\end{proof}

\noindent \macb{Example}:	
\emph{Consider once more statement $S1$ from LU factorization in 
	Figure~\ref{fig:prog_rep}.
	We have \phii{$\boldsymbol{\phi}_0$} = \phii{[i, k]}, 
	\phii{$\boldsymbol{\phi}_1$} = \phii{[i, k]}, and 
	\phii{$\boldsymbol{\phi}_2$} 
	= \phii{[k, k]}. 
	Denote the iteration subdomain for subcomputation 
	$V_h$ as $\mathbf{R}_h = $\linebreak
	$\{$\varvec{$[k^1, i^1]$}, $\dots$, 
	\varvec{$[k^{|V_h|}, 
		i^{|V_h|}]$} $\}$,
	where each variable \var{$k$} and \var{$i$} iterates over its 
	set
	\var{$k^g$}
	$\in \{r_{k,1}, \dots, r_{k,K}\} =$ 
	\range{$R_h^k$} and \var{$i^g$} 
	$\in \{r_{i,1}, \dots, r_{i,I}\} = $ 
	\range{$R_h^i$}, for $g 
	= 
	1, \dots, |V_h|$. Denote the sizes of these 
	sets as $|R_h^k| 
	= K_h$ and $|R_h^i| = I_h$, that is, during $V_h$, 
	variable 
	\var{$k$} takes $K$ different values and \var{$i$} takes $I_h$ 
	different values. For 
	\phii{$\boldsymbol{\phi}_1$}, both iteration 
	variables 
	used are different: \var{k} and \var{i}.
	Therefore, we have (Equation~\ref{eq:a_j_vol}) 
	$|A_1(\mathbf{R}_h)| \le K_h \cdot I_h$. On the other 
	hand, for 
	$\boldsymbol{\phi}_2$, the 
	iteration variable \var{$k$} is used twice. 
	Recall that the access dimension is the minimum number of different 
	iteration 
	variables that uniquely address it 
	(Section~\ref{sec:inputPrograms}), so its dimension is 
	$dim($\arr{$A_2$}$) = 1$ and 
	the only iteration variable needed to uniquely determine 
	$\bm{\phi}_2$ is \var{$k$}.  Therefore,
	$|A_2(\mathbf{R}_h)| \le K_h$.}

\macb{Dominator set.}
Input vertices $A_1, \dots, A_m$ form a dominator 
set of vertices $A_0$,
because any path from graph inputs to any vertex in $A_0$ 
must include \emph{at least} one 
vertex from $A_1, \dots, A_m$.
This is also the 
\emph{minimum} dominator set, because of the disjoint access property
(Section~\ref{sec:inputPrograms}): 
any path from graph inputs to any vertex in $A_0$ can 
include \emph{at most} one 
vertex from  $A_1, \dots, A_m$.

\vspace{1em}
\noindent
\emph{Proof of Lemma~\ref{lma:rectTiling}.}
For subcomputation $V_h$, we have 
$|\bigcup_{j=1}^{m}A_j(\mathbf{R}_h)| \le X$ (by the 
definition 
of an \xpart). Again, by the disjoint access property, we have 
$\forall j_1 \ne j_2: A_{j_1}(\bm{R}_h) \cap A_{j_2}(\bm{R}_h) = \emptyset$.
Therefore, we also have 
$|\bigcup_{j=1}^{m}A_j(\mathbf{R}_h)| = 
\sum_{j=1}^{m}|A_j(\mathbf{R}_h)|$.
We now want to maximize $|V_h|$, that is to find 
$V_{max}$ to 
obtain computational intensity $\rho$
(Lemma~\ref{lma:compIntensityPhi}). 

Now we prove that to maximize $|V_h|$, 
inequalities~\ref{eq:vmax_vol} 
and~\ref{eq:a_j_vol} must be tight (become equalities). 

From proof of Lemma~\ref{lma:vi_bound} it follows that 
$|V_h|$ is maximized 
when iteration vector $\bm{r}$ takes all possible 
combinations of 
iteration variables $r^t_h \in R_h^t$  during $V_h$. But, as we visit each
combination of all $l$ iteration variables, for each access $A_j$
every combination of its $[r_j^1, \dots, r_j^{dim(\bm{\phi}_j)}]$ iteration 
variables is also visited.  
Therefore, for every $j = 1, \dots, m$, each access size 
$|A_j(\bm{R}_h)|$ is maximized  
(Lemma~\ref{lma:projection_bound}), 
as access functions are injective, which implies that for each combination 
of  $[r_j^1, 
\dots, r_j^{dim(\bm{\phi}_j)}]$, there is one access to $A_j$.
{\mbox{$\prod_{t=1}^{l}|R_h^t|$} is then the upper bound on 
	\mbox{$|V_h|$}, and its 
	tightness implies that all bounds on access sizes 
	\mbox{$|A_j(\bm{R}_h)| \le 
		\prod_{k=1}^{dim(\bm{\phi}_j)}|R^k_{h,j}|$} are also tight.
	\qed

\noindent
\Intuition{Lemma~\ref{lma:rectTiling} states that if each iteration variable 
	$r^t, t{=}1, \dots, l$ takes $|R_h^t|$ different values, then there are at 
	most 
	$\prod_{t=1}^{l}|R^t_{h}|$ different iteration vectors $\bm{r}$ which can 
	be 
	formed in $V_h$. Therefore, to maximize $|V_h|$, all combinations of 
	values $r^t$ should be evaluated. On the other hand, this also implies 
	maximization of all access sizes $|A_j(\mathbf{R}_h)| = 
	\prod_{k=1}^{dim(\bm{\phi}_j)}|R^k_{h,j}|$.} 

\subsection{Finding the I/O Lower Bound}
\label{sec:iobound_singlestatement}
Denoting  $V_{max} = \argmax_{V_h \in 
	\mathcal{P}(X)} |V_h|$ the largest subcomputation 
among all valid \mbox{$X$-partitions}, we use 
Lemma~\ref{lma:rectTiling} and combine it with the 
dominator set 
constraint. Note that all access set sizes are strictly positive integers 
$|R^t_{max}| \in \mathbb{N}_+, t = 1, \dots, l$. Otherwise, no computation can 
be performed. However, 
as we 
only want to find the bound on number of I/O operations, we relax 
the 
integer constraints and replace them with $|R^t_{max}| \ge 1$. Then,
we formulate finding $\psi(X)$ (Lemma~\ref{lma:compIntensityPhi}), as the 
optimization problem:

{
\begin{align}
\label{eq:findingX}
\nonumber
\max \prod_{t=1}^{l}|R^t_{max}| & \hspace{2em}\text{s.t.}\\
\nonumber
\sum_{j = 1}^{m} \prod_{k=1}^{dim(\bm{\phi}_j)}|R^k_{max,j}|& \le X \\
\forall 1 \ge t \ge l : |R^t_{max}|& \ge 1 
\end{align} 
}

We then find $|V_{max}| = \psi(X)$ as a function of $X$ using 
Karush– Kuhn–Tucker conditions~\cite{kkt}. Next, we 
solve 
\begin{equation}
\label{eq:drhodx}
\frac{d\frac{\psi(X)}{X - M}}{dX} = 0.
\end{equation}
Denoting $X_0$ as solution to Equation~\ref{eq:drhodx}, we finally obtain
\vspace{-0.5em}
\begin{equation}
\label{eq:Qlowerbound}
Q \ge |V| \frac{(X_0 - M)}{\psi(X_0)}.
\end{equation}

\subsection{Out-degree-one Vertices}
In some cDAGs, every non-input vertex
has 
a certain number $u \ge 0$ of direct predecessors, which are input vertices 
with 
out-degree 1. We can use it to put an additional bound on the computational 
intensity.

\begin{restatable}{lma}{roOnecase}
	\label{lma:rho1case}
	If in a cDAG $G=(V,E)$ every non-input vertex has at least $u$ direct 
	predecessors, with out-degree one, which are graph inputs, then the maximum 
	computational intensity $\rho$ of this cDAG is bounded by $\rho \le 
	\frac{1}{u}$.
\end{restatable}

\begin{proof}
	By the definition of the red-blue pebble game, all inputs start in slow 
	memory, 
	and therefore, have to be loaded. By the assumption on the cDAG, to compute 
	any 
	non-input vertex $v \in V$, at least $u$ input vertices need to have red 
	pebbles placed on them using a load operation.
	Because these vertices do not have any other direct successors (their 
	out-degree is 1), they cannot be used to compute any other non-input vertex 
	$w$. Therefore, each computation of a non-input vertex requires at least 
	$u$ 
	unique input vertices to be loaded. 
\end{proof}

\emph{Example:} Consider Figure~\ref{fig:cdagsRho1}. In~a), each
compute vertex $C[i,j]$ has two input vertices: $A[i,j]$ with out-degree 1, and
$b[j]$ with out-degree $n$, thus $u=1$. As both 
array $A$ and vector $b$ start in the slow memory (having blue pebbles on 
each 
vertex), for each computed vertex from $C$, at least one vertex from $A$ 
has 
to be loaded, therefore $\rho \le 1$. In~b), each computation needs 
two 
out-degree 1 vertices, one from vector $a$ and one from vector $b$, 
resulting 
in $u=2$. Thus, $\rho \le \frac{1}{2}$.
\begin{figure}
	\includegraphics[width=1\columnwidth]{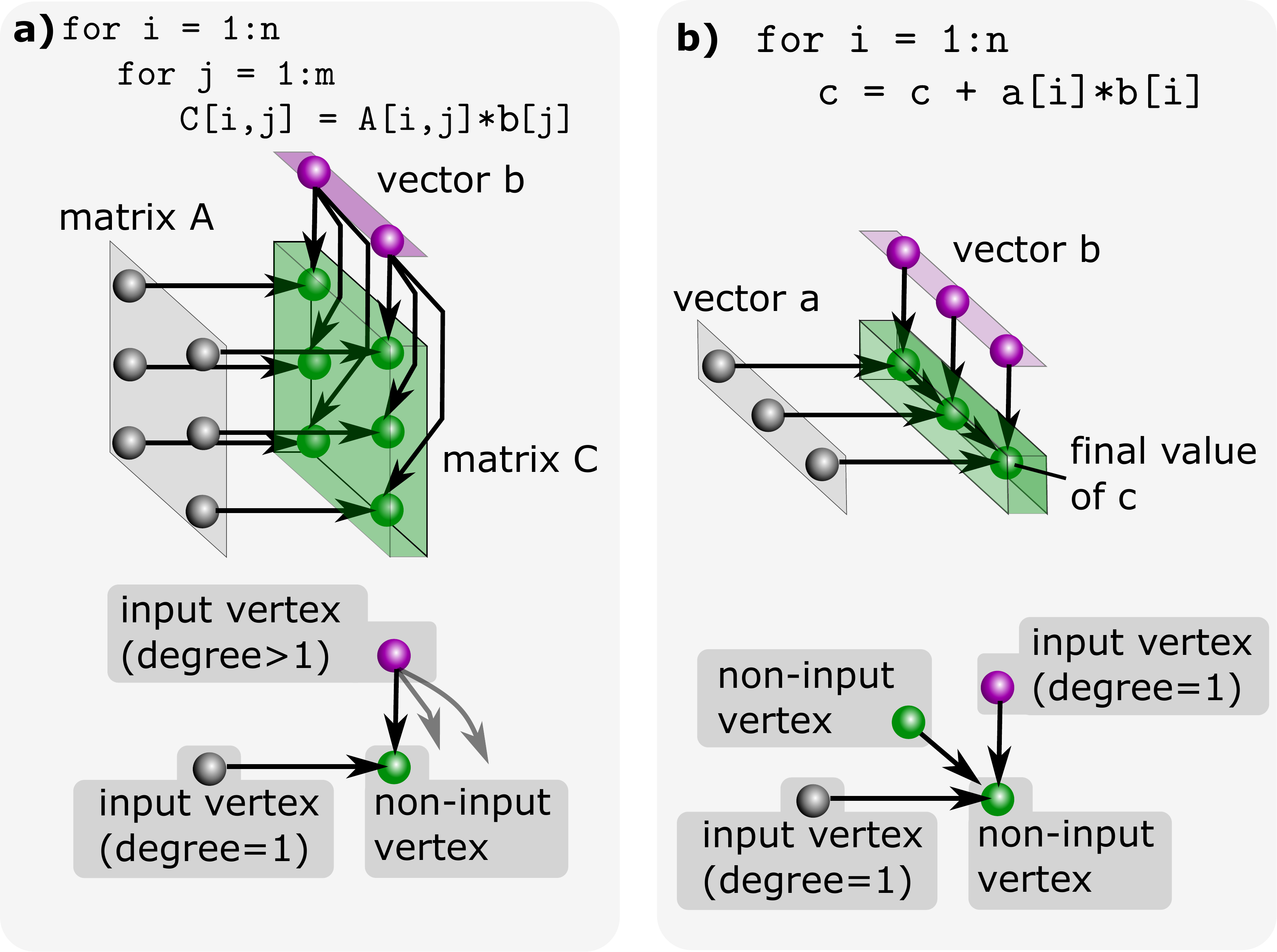}
	\caption{cDAGs with out-degree 1 input vertices. a) $u_a = 1$, $\rho_a 
		\le 
		1$. b) $u_b = 2$, $\rho_b \le \frac{1}{2}$.}
	\label{fig:cdagsRho1}
	
\end{figure}

\Note{We use the above lemma to derive the computational intensity of 
		statement \mbox{$S1$} in LU factorization 
		(Figure{~\ref{fig:prog_rep}}). }

\section{Data Reuse Across Multiple Statements}
\label{sec:mult_statements}
	Almost all computational kernels contain multiple 
	statements connected by data dependencies --- e.g., column update 
	(\mbox{$S1$}) 
	and 
	trailing matrix update (\mbox{$S2$}) in LU factorization 
	(Figure{~\ref{fig:prog_rep}}).
	In this section we examine how these dependencies influence the total I/O 
	cost 
	of a program.

Consider a program containing two statements $S$ and $T$:
{\scriptsize
	\begin{align}
	\nonumber
	\mathbf{for }\hspace{0.5em} &\gamma^1 \in \Gamma^1, \mathbf{for 
	}\hspace{0.5em} \gamma^2 
	\in 
	\Gamma_2(\gamma^1), ... , \mathbf{for }\hspace{0.5em} \gamma^k \in 
	\Gamma_k(\gamma^1, 
	\dots, 
	\gamma^{k-1}) : \\
	\nonumber
	&S : A_0[\boldsymbol{\phi}_0(\mathbf{\gamma})] \leftarrow 
	f(A_1[\boldsymbol{\phi}_1(\mathbf{\gamma})], 
	A_2[\boldsymbol{\phi}_2(\mathbf{\gamma})], \dots, 
	A_m[\boldsymbol{\phi}_m(\mathbf{\gamma})])	 \\
	\nonumber
	\mathbf{for }\hspace{0.5em} &\lambda^1 \in \Lambda^1 : \mathbf{for 
	}\hspace{0.5em} \lambda^2 \in 
	\Lambda^2(\lambda^1),\dots,\mathbf{for }\hspace{0.5em} \lambda^l \in 
	\Lambda^l(\lambda^1, 
	\dots, 
	\lambda^{l-1}) : 
	\\
	\nonumber
	&T : B_0[\boldsymbol{\chi}_0(\mathbf{\lambda})] \leftarrow 
	g(B_1[\boldsymbol{\chi}_1(\mathbf{\lambda})], 
	B_2[\boldsymbol{\chi}_2(\mathbf{\lambda})], \dots, 
	B_n[\boldsymbol{\chi}_n(\mathbf{\lambda})])
	\end{align}
}%

Denote \mbox{$Q_S$} and \mbox{$Q_T$} as I/O costs of statements \mbox{$S$} and 
	\mbox{$T$} if executed separately, and \mbox{$Q_{tot}$} a total I/O cost of 
	the 
	above program.
Assume that there is at least one array that is 
accessed both in $S$ 
and $T$, that is $\exists i, j : A_i = B_j$.
An I/O optimal schedule could take advantage of 
it by possibly fusing 
statements 
$S$ and $T$: once some vertices of $A_i$ are loaded, 
they could be used to compute both $A_0$ (statement $S$) 
and 
$B_0$ (statement $T$), yielding \mbox{$Q_{tot} < Q_S + Q_T$}. However, 
determining explicitly which loops 
should be fused to maximize locality is proven to be 
NP-hard~\cite{loopFusion}. 
Therefore, here we focus 
only on 
the I/O lower bounds, or, in other words, what is the 
maximum possible ``benefit'' of any data reuse between 
statements.

\begin{figure*}
	\includegraphics[width=2.1\columnwidth]{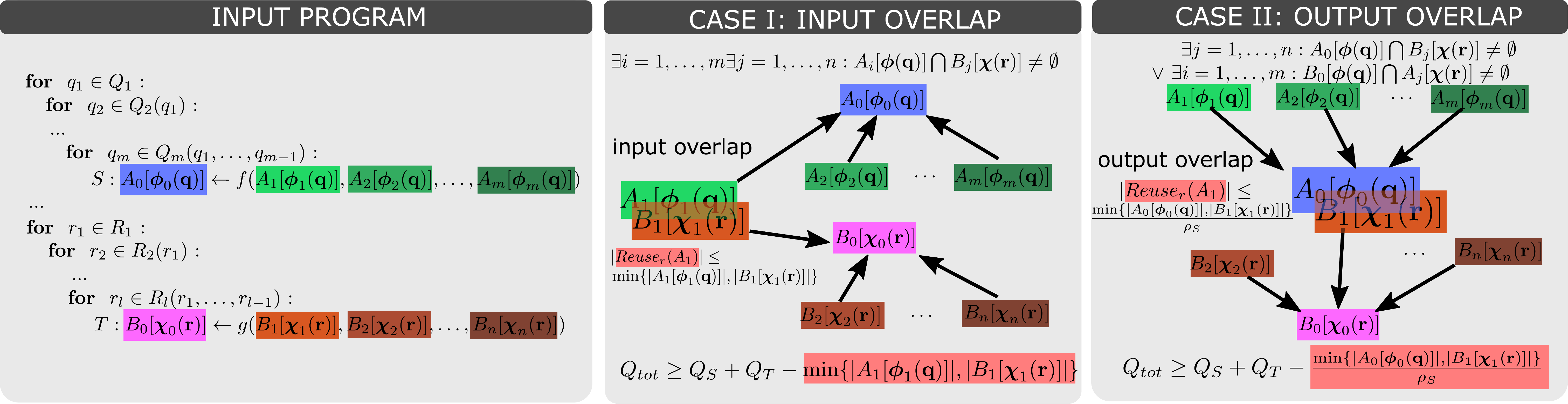}
	\caption{Data reuse across multiple statements.}
	\label{fig:prog_rep2}
	
\end{figure*}

There are two cases in which the data reuse can occur %
	(Figure~\ref{fig:prog_rep2}):
	\textbf{I)} input overlap, where shared 
	arrays are inputs for all statements, \textbf{II)} output 
	overlap, where the output array of one statement is the 
	input array of another. 

\noindent
\textbf{Case I)}. Assume there are \mbox{$w$} statements in the 
program, 
	and there are 
	\mbox{$k$} arrays \mbox{$A_j, j = 1,\dots,k$} which are shared 
between at 
	least 
	two 
	statements. We still evaluate each statement 
	separately, but we will subtract the upper bound on 
	shared loads $Q_{tot} \ge $ \linebreak $\sum_{i 
			= 1}^{w} Q_i - \sum_{j=1}^{k} |Reuse(A_j)|$, where 
	\mbox{$|Reuse(A_j)|$} is the reuse bound on array \mbox{$A_j$}
	(Section{~\ref{sec:reuseSetSize}}).

\noindent	
\textbf{Case II)}. Consider each pair of ``producer-consumer'' 
statements $S$ 
and $T$, that is, the output of $S$ is the input of 
$T$. The I/O lower bound $Q_S$ of statement $S$ does not change due 
to the 
reuse.
 On the other 
hand, 
it may invalidate $Q_T$, as the  
dominator set of $T$ formulated in Section~\ref{sec:access_sizes}
may not be minimum --- inputs of a statement may not be 
graph inputs anymore. For each ``consumer'' statement $T$ we 
reevaluate 
$Q_T' \le Q_T$ using Lemma~\ref{lma:output_reuse}. For a 
program consisting of $w$ statements connected by the 
output 
overlap, we have $Q_{tot} \ge 
\sum_{i = 1}^{w} Q_i'$. Note that for each ``producer'' statement 
$i$, 
$Q_i' = 
Q_i$ (output overlap does not change their I/O lower bound).

\subsection{Case I. Input Reuse and Reuse Size}
\label{sec:reuseSetSize}

Consider two statements $S$ and $T$, which share one input array $A_i$.
Denote $|A_i(\bm{R}_S)|$ the total number of accesses to $A_i$ during the 
I/O optimal execution of a program that contains only statement $S$. 
Analogously, denote $|A_i(\bm{R}_T)|$ for a program containing only $T$.
Define $Reuse(A_i)$ as a number of loads from $A_i$ which are shared 
between statements.

\begin{restatable}{lma}{inputreuse}
	\label{lma:inputreuse}
	The I/O cost of a program containing statements $S$ and $T$ which share the 
	input array $A_i$ is bounded by
	
	$$Q_{tot} \ge Q_{S} + Q_{T} - Reuse(A_i)$$
	
	where $Q_S$, $Q_T$ are the I/O costs of a program containing only statement 
	$S$ or $T$, respectively. Furthermore, we have:
	
	\begin{align}
	\label{eq:reuseSize}
	\nonumber
	Reuse(A_i) \le \min\{|A_i(\bm{R}_S)|, |A_i(\bm{R}_T)|\}
	\end{align}
	{where $|A_i(\bm{R}_S)|$ and $|A_i(\bm{R}_T)|$ are the number of 
		accesses to $A_i$ during the optimal execution of statements 
		$S$ 
		and 
		$T$ separately.}
\end{restatable}

\begin{proof}
	Consider an optimal sequential schedule of a cDAG $G_S$ containing 
	statement 
	$S$ only. For any subcomputation $V_s$ and its associated 
	iteration domain $\bm{R}_s$ its minimum dominator set is
	$\mathit{Dom}(V_s) = \bigcup_{j=1}^m A_j(\bm{R}_s)$. To compute $V_{S}$, 
	at least $\sum_{i=1}^{m}|A_j(\bm{R}_s)| - M$ vertices have to 
	be loaded, as only 
	$M$ vertices can be reused from previous subcomputations.

	We seek if any loads can be avoided in the common schedule if we add 
	statement 
	$T$, denoting its 
	cDAG $G_{S+T}$. 
	Consider a subset  
	$A_i(\bm{R}_x)$ of vertices in $A_i$. 
	
	Consider some subset of vertices in $A_i$ which potentially could be reused 
	and 
	denote it $\Theta_i$. Now 
	denote all vertices in $A_0$ (statement 
	$S$) which depend on any vertex from $\Theta_i$ as $\Theta_S$, and, 
	analogously, set $\Theta_T$ for statement $T$. Now consider these two 
	subsets $\Theta_S$ and $\Theta_T$ separately. If $\Theta_S$ is 
	computed before $\Theta_T$, then it had to load all vertices from 
	$\Theta_i$, 
	avoiding no loads compared to the schedule of $G_S$ only. Now, computation 
	of 
	$\Theta_T$ may take benefit of some vertices from $\Theta_i$, which can 
	still 
	reside 
	in fast memory, avoiding up to $|\Theta_i|$ loads. 
	
	The total number of avoided loads is bounded by the number of loads from 
	$A_i$ 
	which are shared by both $S$ and $T$. Because statement $S$ loads at most  
	$|A_i(\bm{R}_S)|$ vertices from $A_i$ during optimal schedule of $G_S$, and 
	$T$ 
	loads at most $|A_i(\bm{R}_T)|$ of them for $G_T$, the upper bound of 
	shared, 
	and possibly avoided loads is $Reuse(A_i) = \min\{|A_i(\bm{R}_S)|, 
	|A_i(\bm{R}_T)|\}$.
	
\end{proof}

The \textbf{reuse size}
is defined as $Reuse(A_i) = \min\{|A_i(\bm{R}_S)|,$  
$|A_i(\bm{R}_T)|\}$. Now, how to find $|A_i(\bm{R}_S)|$ and $|A_i(\bm{R}_T)|$?

Observe that $|A_i(\bm{R}_S)|$ is a property of $G_S$, that is, the cDAG 
containing statement $S$ only. Denote the I/O optimal schedule 
parameters of $G_S$: $V^S_{max}$, $X^S_0$, and $|A_i(\bm{R}^S_{max}(X^S_0))|$ 
(Section~\ref{sec:iobound_singlestatement}). Similarly, for $G_T$: 
$V^T_{max}$, $X^T_0$, and $|A_i(\bm{R}^T_{max}(X^T_0))|$.
We now derive: 1) at least how 
many subcomputations does the optimal schedule have: $s \ge 
\frac{|V|}{|V_{max}|}$, 2) at least how many accesses to $A_i$ are 
performed 
per optimal subcomputation $|A_i(\bm{R}_{max}(X_0))|$. Then:

\begin{align}
Reuse(A_i) = \min\{&|A_i(\bm{R}^S_{max}(X^S_{0}))| 
\frac{|V^S|}{|V^S_{max}|}, \\
\nonumber
|&A_i(\bm{R}^T_{max}(X^T_{0}))| \frac{|V^T|}{|V^T_{max}|}\}
\end{align}

\noindent \macb{Example:}
Consider the following code:
\vspace{1em}
\begin{lstlisting}[
label={lst:doubleMM}, mathescape=true]
for i = 1:N    for j = 1:N     for k = 1:N
S: D[i,j,k] = A[i,k] * B[k,j]
T: E[i,j,k] = C[i,k] * B[k,j]
end; end; end
\end{lstlisting}

We now derive the I/O lower bound of this program:
\begin{enumerate}
	\item \textbf{statement S}. Denote $I_{h}, 
	J_{h}, K_{h}$ as the number of different values iteration variables 
	$i$, $j$, and $k$ take during the maximal subcomputation $V_{h}$. Then:
	\begin{itemize}
		\item Access sizes (Lemma~\ref{lma:rectTiling}):
		
		$|V^S_{h}| = I_{h}J_{h}K_{h}$, 
		$|A[i,k](\mathbf{R}^S_{h})| = I_{h}K_{h}$, 
		$|B[k,j](\mathbf{R}^S_{h})| = K_{h}J_{h}$
		\item Finding $\psi(X)$ (Optimization problem~\ref{eq:findingX}):
		
		$|V^S_{h}| = \Big(\frac{X}{2}\Big)^2, |A[i,k](\mathbf{R}^S_{h})| = 		
		|B[k,j](\mathbf{R}^S_{h})| = \Big(\frac{X}{2}\Big)^2$
		\item Finding $X_0$ (Equation~\ref{eq:drhodx}):
		
		$X^S_0 = 2M$, $I_h = J_h = K_h = {M}$, $V^S_{h} = 
		M^2$,
		\item Finding the lower bound (Equation~\ref{eq:Qlowerbound}):
		
		$\rho_S = M$, $Q_S = \frac{N^3}{{M}}$
	\end{itemize}
	\item \textbf{statement T}. Analogous to \textbf{S}
	\item $\mathbf{Reuse(}$\texttt{B}$)$
	\begin{itemize}
		\item $\frac{|V^S|}{|V^S_{max}|} = \frac{|V^T|}{|V^T_{max}|}=
		\frac{N^3}{M^2}$,
		$|B[k,j](\mathbf{R^S_{max}})| = KJ = M$, 
		\item $\mathbf{Reuse(}$\texttt{B}$) = \frac{N^3}{M} $
	\end{itemize}
	\item \textbf{I/O lower bound (Lemma~\ref{lma:inputreuse}):} $Q_{tot} = Q_S 
	+ Q_T - Reuse($\texttt{B}$) 
	= 
	\frac{N^3}{M}$
\end{enumerate}

\noindent \textbf{Note:} This bound is attainable by fusing the statements, 
caching $M-1$ elements of matrix $B$, and streaming matrices $A$ and $C$.

\subsection{Case II. Output Reuse and Access Sizes}
\label{sec:output_reuse}

Consider the case where \emph{output} \Ao{$A_{0}$} of 
the statement 
$S$ is 
also the \emph{input} \Bj{$B_{j}$} of statement $T$. 
Consider 
furthermore subcomputation $V_h$ of statement $T$ (and 
its associated iteration domain $\bm{R}_h$).
Any path from the graph 
inputs to vertices in \Bo{$B_0(\bm{R}_h)$} must pass 
through 
vertices in \Bj{$B_{j}(\bm{R}_h)$}. Now
the question is the following:
is there a smaller set of vertices $B_{j}'(\bm{R}_h)$,
$|B_{j}'(\bm{R}_h)| < B_{j}(\bm{R}_h)$ such that 
every 
path from graph inputs to \Bj{$B_{j}(\bm{R}_h)$}
must pass through it?

Denote computational intensity of statement $S$ as 
$\rho_S$. Then we 
state the following lemma:
\begin{restatable}{lma}{outputReuse}
	\label{lma:output_reuse}
	Any dominator set of set $B_{j}(\bm{R}_h)$ must be of size at 
	least $|\mathit{Dom}(B_{j}(\bm{R}_h))| \ge 
	\frac{|B_{j}(\bm{R}_h)|}{\rho_S}$.
\end{restatable}

\begin{proof}
	By Lemma~\ref{lma:compIntensity}, for one loaded vertex, 
	we may 
	compute at most $\rho_S$ vertices of $A_0$. These are 
	also vertices of $B_{j}$. Thus, to compute 
	$|B_{j}(\bm{R}_h)|$ vertices of $B_{j}$, at least 
	$\frac{|B_{j}(\bm{R}_h)|}{\rho_S}$ loads must be 
	performed. We just need to show that at least that many  
	vertices have to be in any dominator set 
	$\mathit{Dom}(B_{j}(\bm{R}_h))$. 
	Now, consider the converse: There is a vertex set
	$D = \mathit{Dom}(B_{j}(\bm{R}_h))$ such that $|D| < 
	\frac{|B_{j}(\bm{R}_h)|}{\rho_S}$. But that would mean, 
	that we could potentially compute all $|B_{j}(\bm{R}_h)|$ 
	vertices by only loading $|D|$ vertices,   
	violating Lemma~\ref{lma:compIntensity}.
\end{proof}

\begin{corollary}
	\label{cor:output}
	Combining Lemmas~\ref{lma:output_reuse} 
	and~\ref{lma:rectTiling},
	the data access size of $|B_{j}(\bm{R}_h)|$ during 
	subcomputation $V_h$ is
	\begin{equation}
	\label{eq:inputOutputReuseSize}
	|B_{j}(\bm{R}_h)| \ge 
	\frac{\prod_{k=1}^{dim(\bm{\phi}_j)}|R_{h,j}^k|}{\rho_S}.
	\end{equation}
\end{corollary}

\noindent \macb{Example (Modified Matrix 
	Multiplication~\cite{general_arrays}):}

\begin{lstlisting}[
label={lst:comp_onthefly}, mathescape=true]
  for i = 1:N
     for j = 1:N
S:      A[i,j] = $e^{2\pi \sqrt{-1} (i - 1)(j-1)/N}$
        for k = 1:N
T:         C[i,j] = A[i,k]*B[k,j] + C[i,j]
end; end; end
\end{lstlisting}

Consider the code above.
Statement $S$ does not have any input arrays (we assume that iteration 
variables $i$ and $j$ are always loaded in the registers. Therefore, there 
are no loads performed during the execution of $S$, {so \mbox{$\rho_S 
		\rightarrow \infty$} for large $N$}. Statement $T$, on the other hand, 
		if 
executed separately, would perform at least $Q_T \ge \frac{2N^3}{\sqrt{M}}$ 
loads. 
However, using Corollary~\ref{cor:output}, {we obtain access size 
	\mbox{$|A_1(\bm{R}_h)| \ge \frac{|R_h^i| |R_h^k|}{\rho_S} \ge 0$}}, 
and the combined bound is 
$Q_{T+S} \ge \frac{N^3}{M}$. This bound is tight, as the I/O 
optimal 
schedule would cache $M-1$ vertices of $C$, and for each loaded vertex of 
$B$ would compute $M-1$ new vertices of $C$.

\section{Deriving Parallel I/O Lower Bounds}
\label{sec:parredblue}

{We now establish how our method applies to a parallel machine with 
$P$ 
	processors (Section{~\ref{sec:machineModel}})}.
Each processor $p_i$ owns its private fast memory 
which can hold up to $M$ words, represented in the cDAG as 
$M$ red vertices with $p_i$'s ``hue''. Red vertices of 
different hues (belonging to different processors) cannot be 
shared between them, but any number of different red 
pebbles may be placed on one vertex. 

All the standard red-blue pebble game rules apply with the 
following modifications:
\begin{enumerate}[leftmargin=*]
	\item \textbf{compute} if all direct predecessors of 
	vertex 
	$v$ have red pebbles of $p_i$'s hue placed on them, one 
	can place a red pebble of $p_i$'s hue on $v$ (no sharing 
	of red pebbles between processors),
	\item \textbf{load} if a vertex $v$ has \emph{any} 
	pebble 
	placed on them,
	a red pebble of 
	any other hue may be placed on a vertex.
\end{enumerate}

From this game definition, it follows that from a perspective 
of a single processor $p_i$, any data is either local (the 
corresponding vertex has $p_i$'s red pebble placed on it), or remote, 
{without a distinction on the remote location (remote access cost is 
	uniform).}

\begin{restatable}{lma}{parallelIO}
	\label{lma:parallelIO}
	The minimum number of I/O operations in a parallel 
	red-blue pebble game,
	played on a cDAG with $|V|$ vertices with $P$ processors 
	each equipped with $M$ red pebbles, is $Q \ge \frac{|V|}{P \cdot \rho}$,
	where $\rho$ is the maximum computational intensity independent of $P$
	(Lemma~\ref{lma:compIntensity}).
\end{restatable}

\begin{proof}
	Following the analysis of Section~\ref{sec:boundsSingleStatement} and 
	the parallel machine 
	model (Section~\ref{sec:parredblue}), the computational 
	intensity $\rho$ is independent of a number of parallel 
	processors - it is solely a property of a cDAG and private 
	fast memory size $M$. Therefore, following 
	Lemma~\ref{lma:compIntensity}, what changes with $P$ is the 
	volume of computation $|V|$, as now at least one processor will 
	compute at least $|V_p| = \frac{|V|}{P}$ vertices. 
	By 
	the definition of the computational intensity, the 
	minimum number of I/O operations required to pebble these 
	$|V_p|$ vertices is $\frac{|V_p|}{\rho}$.
\end{proof}

\section{\hspace{-0.0em}Bounds of Parallel LU Factorization}
\label{sec:lu_lowerbound}

In the previous sections, we have analyzed all components of 
the LU factorization algorithm (Figure~\ref{fig:prog_rep}) 
separately. We now provide a full, end-to-end derivation of its  
parallel I/O lower bound using 
our method. Previously, Olivry et 
al.~\cite{olivry2020automated} reported a lower bound for a 
sequential machine $\frac{2}{3}\frac{N^3}{\sqrt{M}}$. To the 
best of our knowledge, this is the first parallel result for 
this algorithm.

Recall that the algorithm  contains two statements:

\noindent
\textbf{S1: \texttt{A[i,k] = A[i,k]/A[k,k]}}

Denote $|R^k_{h}| = K_h$, $|R^i_{h}| = I_h$. Then, we 
have the following (Lemma~\ref{lma:rectTiling}):
\begin{itemize}[leftmargin=0.75em]
	\item $|V_{h}| = K_h I_h$
	\item $|A_1(\mathbf{R}_h)| = K_h I_h$; 
	\hspace{1em} $|A_2(\mathbf{R}_h)| = K_h$
	\item $|\mathit{Dom}(V_h)| = |A_1(\mathbf{R}_h)| + 
	|A_2(\mathbf{R}_h)| = 
	K_h I_h + K_h$
\end{itemize}

We then solve the optimization problem from 
Section~\ref{sec:iobound_singlestatement}:
\begin{align}
\nonumber
\max \text{\hspace{0.5em}} K_h I_h, & \hspace{2em}\text{s.t.}\\
\nonumber
K_h I_h + K_h& \le X \\
\nonumber
I_h& \ge 1 \\
\nonumber
K_h& \ge 1
\end{align}

\begin{figure*}
	\includegraphics[width=1.9\columnwidth]{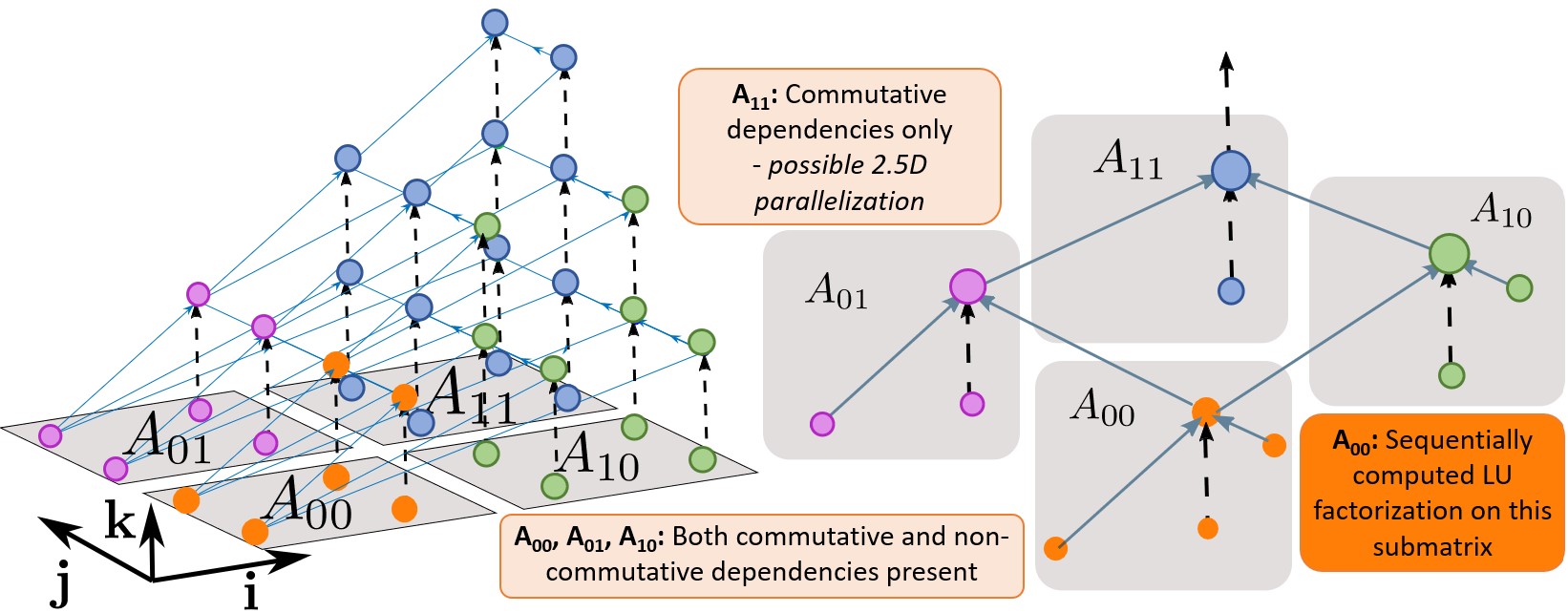}
	\caption{LU cDAG for \mbox{$n=4$} together with the 
		logical 
		decomposition to 
		\mbox{$A_{00}, A_{10}, A_{01}$}, and \mbox{$A_{11}$}. 
		Dashed arrows 
		represent 
		commutative 
		dependencies (reduction of a value). Solid arrows 
		represent non-commutative operations, so 
		any 
		parallel pebbling has to respect the induced order 
		(e.g., no vertex 
		in 
		\mbox{$A_{11}$} can be pebbled before \mbox{$A_{00}$} 
		is pebbled).}
	\label{fig:cdags}
\end{figure*}

Which gives $|V_{max}| = \psi(X) = X - 1$ for $K_h = 1$ and $I_h = X - 1$.
Then $\rho(X) = \frac{|V_{max}|}{X-M} = \frac{X - 1}{X - M}$. However, 
because 
$A_1$ has out-degree 1, we use the bound from 
Lemma~\ref{lma:rho1case}: 
$\rho_{S1} \le 1$. Preserving the correctness of I/O lower bounds,
we use its upper 
bound $\rho_{S1} = 1$.

Finally, we calculate total number of vertices in statement \textbf{S1}: 
$|V_{S1}| = \sum_{k=1}^N(N-k-1) = 
\frac{N(N-1)}{2}$ and conclude that $Q_{S1} \ge 
\frac{|V_1|}{\rho_1} = 
\frac{N(N-1)}{2}$ (Lemma~\ref{lma:compIntensity}).

\vspace{1em}
\noindent
\textbf{S2: \texttt{A[i,j] = A[i,j] - A[i,k]*A[k,j]}}

{
	Denote	\mbox{$|R^k_{h}| = K_h$}, \mbox{$|R^i_{h}| = I_h$}, 
	\mbox{$|R^j_{h}| = 
		J_h$}. 
	Observe that there is  
	an output reuse (Section{~\ref{sec:output_reuse}}) of \texttt{A[i,k]} 
	between 
	statements \mbox{$S1$} (as $A_0$) and 
	\mbox{$S2$} (as $A_2$)
	. We therefore have the 
	access size in statement \textbf{S2:} 
	\mbox{$|A_2(\mathbf{R}_{S2})| = 
		\frac{I_h K_h}{\rho_{S1}} =I_h K_h$}
	(Equation{~\ref{eq:inputOutputReuseSize}}). Note that in this case, where 
	the 
	computational intensity is \mbox{$\rho_{S1} \le 1$},
	the output reuse does not 
	change 
	the 
	access size \mbox{$|A_2(\mathbf{R}_{S2})|$} of statement \mbox{$S2$}. This 
	follows the 
	intuition that it is not beneficial to recompute vertices if the 
	recomputation 
	cost is not lower than loading it from the memory. 
	
	The remaining steps of the I/O lower bound analysis are similar to 
	\mbox{$S1$}.
	We then obtain \mbox{$\rho_{S2} = \frac{\sqrt{M}}{2}$}, \mbox{$|V_{S2}| = 
		\frac{N^3}{3} - N^2 + \frac{2N}{3}$} and finally  \mbox{$Q_{S2} \ge 
		\frac{2N^3 
			- 
			6N^2 + 4N}{3\sqrt{M}}$}.}
The I/O lower bound of the full LU factorization is 
therefore: 
$$	Q_{LU} \ge Q_1 + Q_2 \ge \frac{2N^3 - 6N^2 + 
	4N}{3\sqrt{M}} + \frac{N(N-1)}{2}$$

Using Lemma~\ref{lma:parallelIO} we have the parallel I/O 
lower bound
{\scriptsize
	$${ Q_{P,LU} \ge \frac{2N^3 - 6N^2 + 4N}{3P\sqrt{M}} + 
		\frac{N(N-1)}{2p} = \frac{2N^3}{3P\sqrt{M}} + 
		\mathcal{O}\Big(\frac{N^2}{P}\Big)},$$}
\noindent
which is one of the main contributions of our work.

\section{\conflux}
\label{sec:conflux}

In this section we present
	{\conflux} --- a 
	near \emph{Communication Optimal LU	factorization using 
	{\xparting}}.

\subsection{LU Dependencies and Parallelization}
\label{sec:par_lu_no_pivot}

Due to the dependency structure of LU, the input matrix is often 
divided recursively into four submatrices $A_{00}$, $A_{10}$, 
$A_{01}$, and $A_{11}$~\cite{LUdongarra, 2.5DLU}.
{Arithmetic operations performed in LU create 
	non-commutative dependencies (Figure{~\ref{fig:cdags}}) between 
	vertices 
	in \mbox{$A_{00}$} (LU 
	factorization 
	of the top-left corner of the matrix), \mbox{$A_{10}$}, and \mbox{$A_{01}$} 
	(triangular solve 
	of vertical and top panels of the matrix).} Only $A_{11}$ (Schur complement 
update) has no such dependencies, and may therefore be efficiently 
parallelized in the reduction dimension. Our parallel algorithm utilizes 
this fact and applies different strategies for different parts.
Its high-level summary is presented in Algorithm~\ref{alg:conflux}.

\begin{algorithm}
	\footnotesize
	\caption{\conflux}
	\label{alg:conflux}
	\begin{algorithmic}
		\State $A_t \gets A$
		\For{$t = 1, \dots, \frac{N}{v}$} 
		\State \textbf{1.} Reduce next block column
		\Comment{Cost: 
			$\frac{(N-t\cdot 
				v)\cdot v\cdot M}{N^2}$}
		\State \textbf{2.} $TournPivot(A_t)$ 
		\Comment{Cost: 
			$v^2\left\lceil\log(\frac{N}{\sqrt{M}})\right 
			\rceil$}  
		\State \textbf{3.} Scatter computed $A_{00}$ and $v$ 
		pivot rows 
		\Comment{Cost: $v^2 + v$}
		\State \textbf{4.} Scatter $A_{10}$ 
		\Comment{Cost: $\frac{(N-t\cdot 
				v)v}{P}$}
		\State \textbf{5.} Reduce $v$ pivot rows 
		\Comment{Cost: 
			$\frac{(N-t\cdot 
				v)\cdot v\cdot M}{N^2}$}
		\State \textbf{6.} Scatter $A_{01}$ 
		\Comment{Cost: $\frac{(N-t\cdot 
				v)v}{P}$}
		\State \textbf{7.} $FactorizeA_{10}(A_t)$ \Comment{1D 
			parallel., 
			block-row}
		\State \textbf{8.} Send data from panel $A_{10}$ 
		\Comment{Cost: 
			$\frac{(N-t\cdot 
				v)N\cdot v}{P\sqrt{M}}$}
		\State \textbf{9.} $FactorizeA_{01}(A_t)$ \Comment{1D 
			parallel., block-column} 
		\State \textbf{10.} Send data from panel $A_{01}$ 
		\Comment{Cost: 
			$\frac{(N-t\cdot 
				v)N\cdot v}{P\sqrt{M}}$}
		\State \textbf{11.} $FactorizeA_{11}(A_t)$ 
		\Comment{2.5D parallel.} 
		\State $A_t \gets 
		A_t[rows, v:end]$
		\EndFor			
	\end{algorithmic}
\end{algorithm}

\subsection{Computation Routines}
The computation is performed in $\frac{N}{v}$ steps, 
where $v$ is a tunable blocking parameter.
In each step, only submatrix $A_t$ of input matrix $A$ is updated. Initially, 
$A_t$ 
is set to $A$. $A_t$ is further decomposed to four 
submatrices $A_{00}$, $A_{10}$, $A_{01}$, and $A_{11}$ which are updated by 
routines $TournPivot$, $FactorizeA_{10}$, $FactorizeA_{01}$, and 
$FactorizeA_{11}$ 
(see Figure~\ref{fig:lu_decomp}):

\begin{itemize}[leftmargin=0.75em]
	\item {\mbox{$\bm{A_{00}}$}. This \mbox{$v \times v$} submatrix contains 
		first \mbox{$v$} 
		elements of current \mbox{$v$} pivot rows. It is computed during 
		\mbox{$TournPivot$}, and 
		as 
		it is required to compute \mbox{$A_{10}$} and \mbox{$A_{01}$}, it is 
		redundantly copied 
		to all processors.}
	\item {\mbox{$\bm{A_{10}}$} and \mbox{$\bm{A_{01}}$}. 
		Submatrices \mbox{$A_{10}$} and \mbox{$A_{01}$} of sizes 
		\mbox{$(N-t\cdot 
			v) \times v$} and  \mbox{$v 
			\times (N-t\cdot v)$} are 
		distributed using 1D 
		decomposition among all processors. They are updated using a triangular 
		solve. 1D decomposition 
		guarantees 
		that there are no dependencies between processors, so no communication 
		or 
		synchronization is performed during computation (\mbox{$A_{00}$} is 
		already 
		owned 
		by every 
		processor).}
	\item 	{\mbox{$\bm{A_{11}}$} This \mbox{$(N-t\cdot v) \times (N-t\cdot 
			v)$} 
		submatrix
		is distributed using 2.5D, block-cyclic distribution 
		(Figure{~\ref{fig:lu_decomp}}). 
		First, updated submatrices \mbox{$A_{10}$} and \mbox{$A_{01}$} are 
		broadcast among the 
		processes. Then, \mbox{$A_{11}$} (Shur complement) is updated.
		Finally, the first block column and \mbox{$v$} chosen pivot rows are 
		reduced, 
		which will 
		form \mbox{$A_{10}$} and \mbox{$A_{01}$} in the next iteration.}
\end{itemize}

\noindent \macb{Blocking parameter $\bm{v}$}. The minimum size of each block is 
the number of processor layers in the reduction dimension $v \ge c = 
\frac{PM}{N^2}$. However, to secure high performance, this value should also be 
adjusted to hardware parameters of a given machine (e.g., vector length, 
prefetch distance of a CPU, or warp size of a GPU). Throughout the analysis, we 
assume that $v = a \cdot \frac{PM}{N^2}$ for some small 
constant $a$.  

\begin{figure*}
	\begin{center}
		\includegraphics[width=1.9\columnwidth]{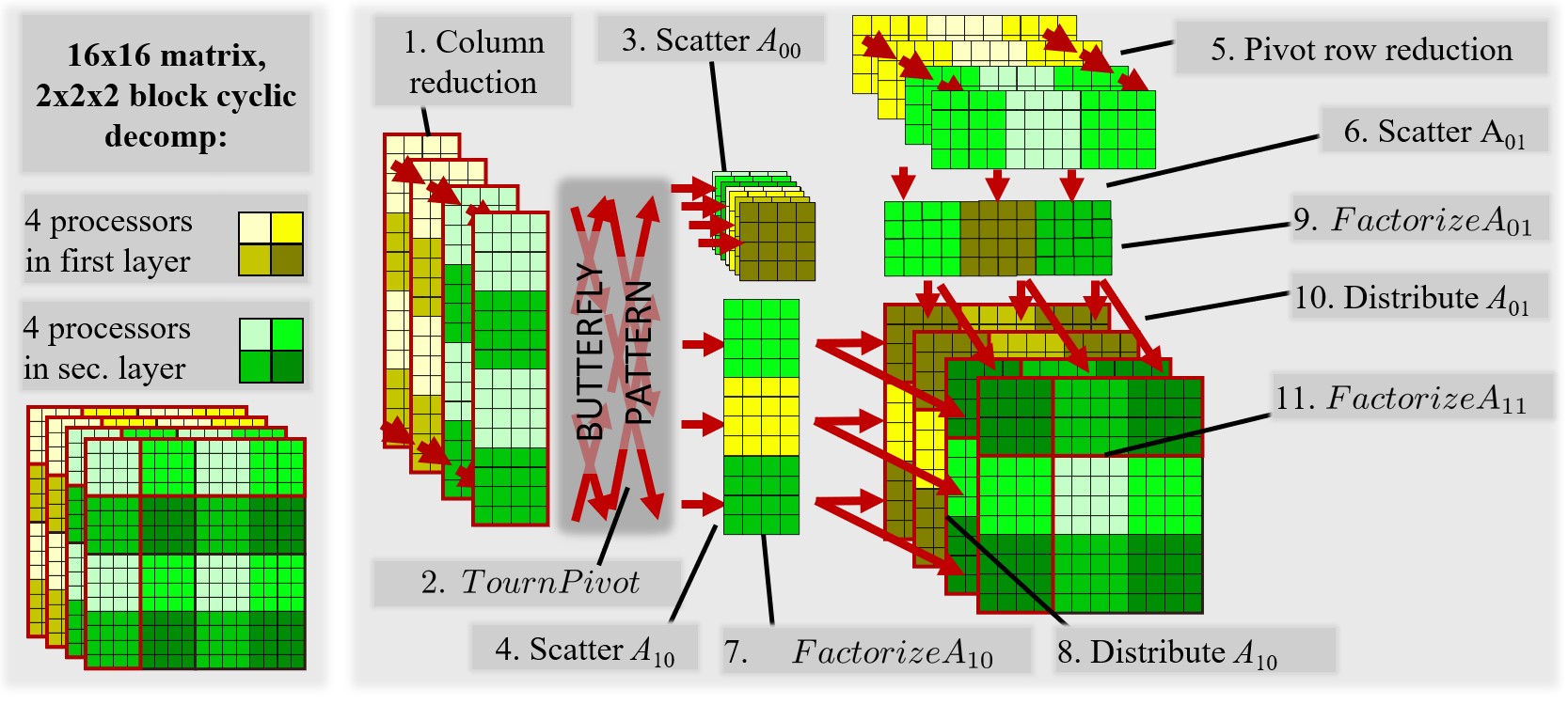}
	\end{center}
	
	\caption{CO$nf$LUX parallel decomposition for $P=8$ 
	processors decomposed into $2 \times 2 \times 2$ grid, 
	together with indicated steps of   
		Algorithm~\ref{alg:conflux}.}
	
	\label{fig:lu_decomp}
\end{figure*}

\subsection{Pivoting}
\label{sec:par_lu_pivot}

Our pivoting strategy differs from state-of-the-art 
block~\cite{lapack}, tile~\cite{plasma}, or recursive~\cite{recursivePivoting} 
pivoting
approaches in two aspects: 
\begin{itemize}[leftmargin=*]
	\item To minimize I/O, we do not swap pivot rows. Instead, we keep track 
	which 
	rows were chosen as pivots and we use masks to update remaining rows.
	\item To reduce latency, we take advantage of our derived blocks 
	and use tournament pivoting~\cite{tourn_pivot}.
\end{itemize}

\noindent
The tournament pivoting finds $v$ 
pivot rows in each step, which are then used to
mask which rows will form the new $A_{01}$ 
and then filter the non-processed row in the next 
step.

\begin{table*}[h]
	\setlength{\tabcolsep}{2pt}
	\renewcommand{\arraystretch}{0.7}
	\centering
	\footnotesize
	
	\begin{tabular}{p{3.2cm}p{3cm}p{3cm}p{4cm}p{4cm}}
		\toprule
		& \textbf{LibSci}
		& \textbf{SLATE} 			
		& \textbf{CANDMC} 			
		& \textbf{\conflux}\\
		\midrule
		\textbf{Decomposition}
		& 
		2D, panel decomp.
		&
		2D, block decomp.
		&
		Nested 2.5D, block decomp.
		& 
		1D / 2.5D, block decomp.
		\\
		\begin{tabular}{l}
			\hspace{-0.6em}
			\textbf{Block size}
		\end{tabular}
		&
		\begin{tabular}{ll}
			\begin{tabular}{l}
				\hspace{-0.5em}
				\includegraphics[width=0.047 \textwidth]
				{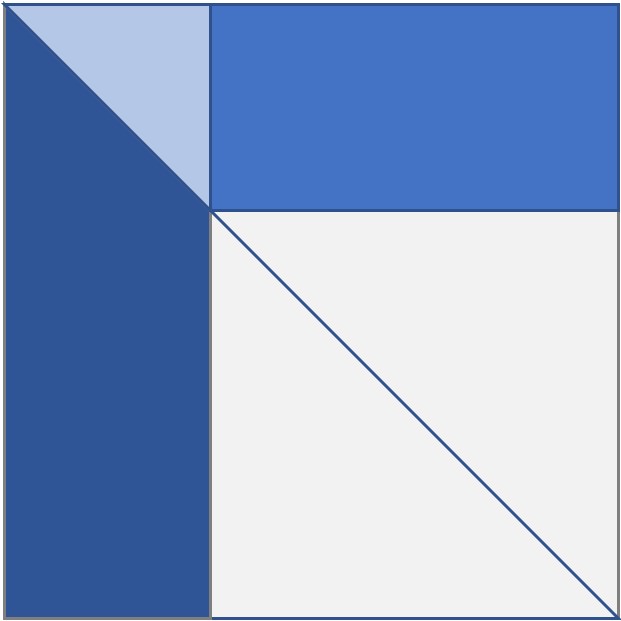}
			\end{tabular}
			&
			\begin{tabular}{l}
				user-specified
			\end{tabular}
		\end{tabular}
		&
		\begin{tabular}{ll}
			\begin{tabular}{l}
				\hspace{-0.5em}
				\includegraphics[width=0.047 \textwidth]
				{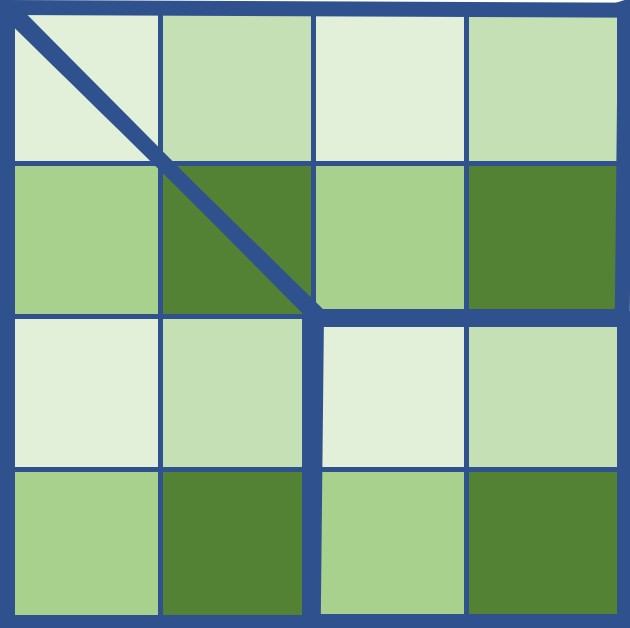}
			\end{tabular}
			&
			\begin{tabular}{l}
				user-specified,\\
				(default 16)
			\end{tabular}
		\end{tabular}
		&
		\begin{tabular}{ll}
			\begin{tabular}{l}
				\hspace{-0.5em}
				\includegraphics[width=0.047 \textwidth]
				{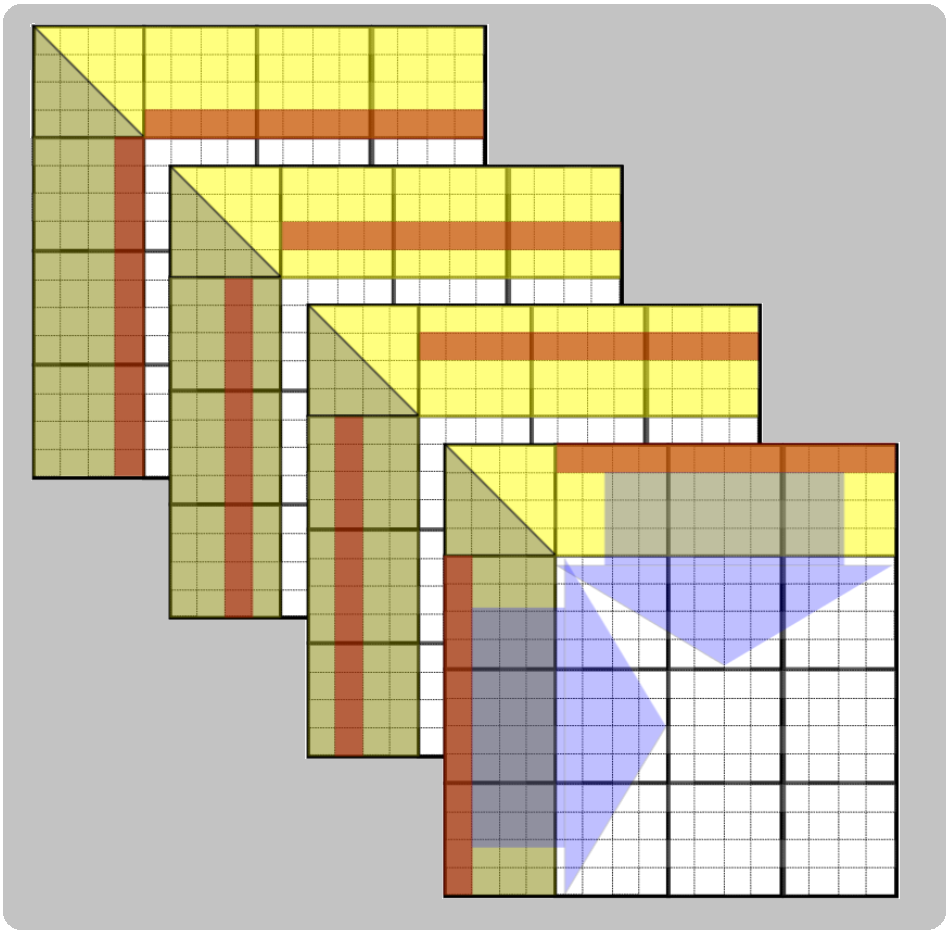}
			\end{tabular}
			&
			\begin{tabular}{l}
				$\frac{N^3}{P\cdot M}$ , $\frac{N^2}{P\sqrt{M}}$
			\end{tabular}
		\end{tabular}
		&
		\begin{tabular}{ll}
			\begin{tabular}{l} 
				\hspace{-0.5em}
				\includegraphics[width=0.052 \textwidth]
				{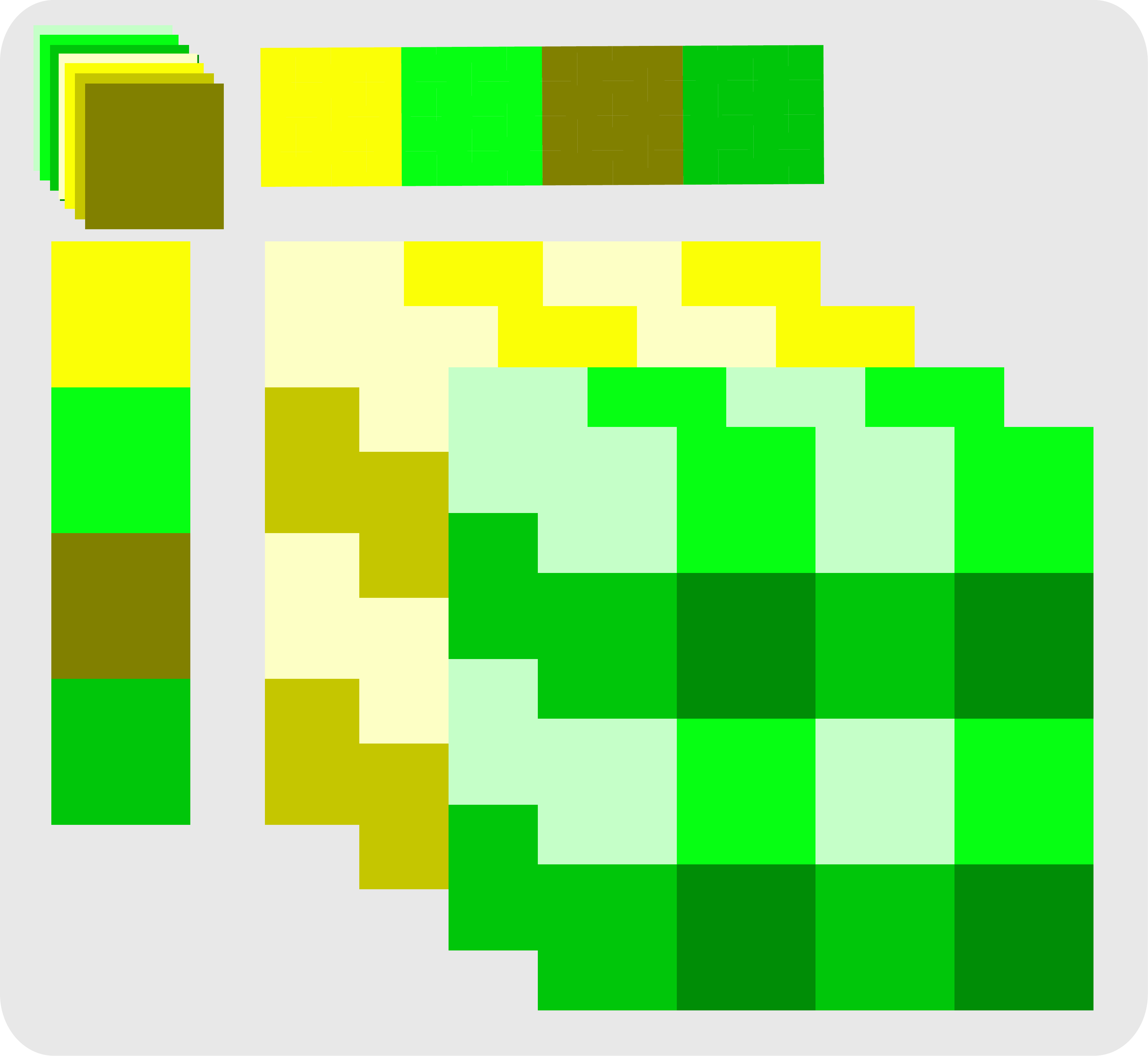}
			\end{tabular}
			&
			\begin{tabular}{l}
				tunable, $\ge \frac{P\cdot M}{N^2}$
			\end{tabular}
		\end{tabular}
		\\
		\textbf{User param. required} & 
		yes~\faThumbsDown & 
		no~\faThumbsOUp & 
		no~\faThumbsOUp & 
		no~\faThumbsOUp 	\vspace{1em}\\
		\textbf{Parallel I/O cost} 
		& 
		$\frac{N^2}{\sqrt{P}} + 
		\mathcal{O}\Big(\frac{N^2}{P}\Big)$
		&
		$\frac{N^2}{\sqrt{P}} + 
		\mathcal{O}\Big(\frac{N^2}{P}\Big)$
		&
		$\frac{5N^3}{P\sqrt{M}} + 
		\mathcal{O}\Big(\frac{N^2}{P\sqrt{M}}\Big)$~\cite{2.5DLU}
		&
		$\frac{N^3}{P\sqrt{M}} + 
		\mathcal{O}\Big(\frac{N^2}{P\sqrt{M}}\Big)$
		\\
		\midrule
		\multicolumn{5}{c}{\textbf{Total comm. volume for $N=4$,$096$ 
				measured/modeled [GB] (prediction \%)}} \\
		$\bm{P=64}$ &
		1.17 / 1.21 (102\%) & 1.18 / 1.21 (102\%)  &  2.5 / 4.9 
		(196\%)& 1.11 / 1.08 (97\%) \\
		$\bm{P=1,024}$ & 4.45 / 4.43 (99\%) & 4.35 / 4.43 (102\%) & 9.3 / 12.13 
		(130\%) 
		& 3.13 / 3.07 (98\%) \\
		\midrule
		\multicolumn{5}{c}{\textbf{Total comm. volume for $N=16$,$384$ 
				measured/modeled [GB] (prediction \%)}} \\
		$\bm{P=64}$ & 18.79 / 19.33 (103\%) &  18.84 / 19.33 (102\%)  &
		39.8 / 78.74 (197\%) & 17.61 / 17.19 (98\%) \\
		$\bm{P=1,024}$ & 70.91 / 70.87 (99.9\%) & 71.1 / 70.87 (99.7\%)  & 144 
		/ 194.09 (135\%) & 45.42 / 44.77 (98\%) \\
		\bottomrule
	\end{tabular}
	\caption{
		\textmd{{Classification and I/O cost models of the 
				measured LU 
				factorization implementations. CANDMC model is taken from the 
				authors~\cite{2.5DLU}. Due to the space constraints, we omit 
				the lower 
				order terms of the models.
		}}
	}
	
	\label{tab:results}
\end{table*}

\noindent \macb{Tournament Pivoting} is shown to be as stable as partial 
pivoting~\cite{tourn_pivot}, which might be an issue for, e.g., incremental 
pivoting~\cite{incrementalPivoting}.
On the other hand, it reduces the $\mathcal{O}(N)$ latency cost of the partial 
pivoting, which requires step-by-step 
column reduction to find consecutive 
pivots,  to $\mathcal{O}\big(\frac{N}{v}\big)$, where $v$ is the tunable block 
size parameter.

\noindent \macb{Row Swapping vs. Row Masking}.
To achieve close to optimal I/O cost, we use 2.5D decomposition. This, 
however, implies that in the presence of extra memory, the matrix data is 
replicated $\frac{PM}{N^2}$ times. 
This increases the row swapping cost
from  
$\mathcal{O}\big(\frac{N^2}{P})$ to 
$\mathcal{O}\big(\frac{N^3}{P\sqrt{M}}\big)$ which asymptotically matches 
the I/O lower bound of the entire factorization. Performing row swapping would 
then increase the 
constant term of the leading factor of the algorithm from 
$\frac{N^3}{P\sqrt{M}}$ to $\frac{2N^3}{P\sqrt{M}}$.
To keep the I/O cost of our algorithm as low as possible, instead of performing 
row-swapping, we only propagate pivot row indices. When the tournament pivoting 
finds the $v$ pivot rows, they are broadcast to all processors with only 
${v}$ cost per step. 

\noindent \macb{Pivoting in {\conflux}}.
In each step \mbox{$t$} of the outer loop (line 1 in 
Algorithm~\ref{alg:conflux}),
\mbox{$\frac{N}{\sqrt{M}}$} processors perform a tournament pivoting routine 
using 
a butterfly communication pattern~\cite{butterfly}. Each processor owns 
$\sqrt{M}\frac{N - vt}{N}$ rows, among which it chooses $v$ local candidate 
pivots. Then, final pivots are chosen in  \mbox{$\log(\frac{N}{\sqrt{M}})$}  of 
``playoff-like'' tournament rounds, after which all  
\mbox{$\frac{N}{\sqrt{M}}$} 
processors own both $v$ pivot row indices and already factored new $A_{00}$.
This result is distributed to all remaining processors (line 2).
Pivot row indices are then used to determine which processors participate in 
the reduction of current  \mbox{$A_{01}$} (line 4).
Then, the new $A_t$ is formed by masking currently chosen rows $A_t \gets 
A_t[rows, v:end]$ (Line 12).

\subsection{I/O cost of 
		\conflux}

{We now prove the I/O cost of 
	\conflux{}, which is only a factor of $\frac{1}{3}$ higher than the lower 
	bound for large $N$.}

\begin{lma}
	{
		The total I/O cost of \conflux, presented in 
		Algorithm{~\ref{alg:conflux}}, is 
		\mbox{$Q_{COnfLUX} = \frac{N^3}{P\sqrt{M}} + 
			\mathcal{O}\left(\frac{N^2}{P} \right)$}.}
\end{lma}

\begin{proof}
	{
		We assume that the input matrix  \mbox{$A$} is already 
		distributed in the block cyclic layout imposed by the 
		algorithm. Otherwise, any data reshuffling imposes only a
		\mbox{$\Omega\big(\frac{N^2}{P}\big)$} cost, which does not contribute 
		to 
		the 
		leading order term.
		We first derive the cost of a single iteration~\mbox{$t$} of the 
		main loop of the algorithm, proving its 
		cost to be \mbox{$Q_{step}(t) = \frac{2Nv(N-tv)}{P\sqrt{M}} + 
			\mathcal{O}\left(\frac{Nv}{P} \right)$}.
		Then, the total cost after  \mbox{$\frac{N}{v}$} iterations is:}
	
	{\footnotesize
	$$Q_{COnfLUX} = \sum_{t=1}^{\frac{N}{v}}Q_{step}(t) = 
	\frac{N^3}{P\sqrt{M}} 
	+ 
	\mathcal{O}\left(\frac{N^2}{P} \right). $$
	}

		We denote 
		\mbox{$P1 = \frac{N^2}{M}$ and $c = \frac{PM}{N^2}$}.
		\mbox{$P$} processors are decomposed in the 3D grid 
		\mbox{$[\sqrt{P1}, \sqrt{P1},c]$}. We refer to all processors 
		which share the same second and third coordinate as 
		\mbox{$[:, j, k]$}. We now examine each of 11 
		steps of Algorithm 2.
		
		{\noindent}
		\textbf{Step 1.}  \mbox{$[:, t \mod \sqrt{P1}, t \mod c]$} 
		processors perform the tournament pivoting. Every 
		processor owns first \mbox{$v$} elements of \mbox{$N - (t-1)v$} 
		rows, among which they choose the next  \mbox{$v$} pivots. 
		First, they locally perform the LUP decomposition to 
		choose local  \mbox{$v$} candidate rows. Then, in 
		\mbox{$\lceil \log_2(\sqrt{P1})\rceil$} rounds they exchange 
		\mbox{$v \times v$} blocks to decide on the final pivots. After the 
		exchange, these processors also hold the factorized 
		submatrix  \mbox{$A_{00}$}. 
		\emph{I/O cost per proc.:}  \mbox{$v^2 \lceil 
		\log_2(\sqrt{P1})\rceil$}. 
		
		{\noindent}
		\textbf{Steps 2, 3, 5.} Factorized 
		\mbox{$A_{00}$}
		and  \mbox{$v$} pivot row indices 
		are broadcast. First  \mbox{$v$} 
		columns and \mbox{$v$} 
		pivot rows
		are scattered to all  \mbox{$P$}.
		\emph{I/O cost 
			per proc.:}  \mbox{$v^2 + v + \frac{2(N-tv)v}{P}$}. 
		
		{\noindent}
		\textbf{Steps 4 and 11.} Reduce \mbox{$v$} columns and \mbox{$v$} 
		pivot rows. With high probability, pivots are evenly distributed among 
		all 
		processors.
		 There are  \mbox{$c$} layers to reduce, 
		each 
		of size 
		\mbox{$(N-tv)v$}.
		\emph{I/O cost per proc.:} 
		\mbox{$\frac{(N-tv)vc}{P} = \frac{2(N-tv)vM}{N^2}$}. 
		
		{\noindent}
		\textbf{Steps 6, 8, 10.} The updates \mbox{$FactorizeA_{10}$}, 
		\mbox{$FactorizeA_{01}$}, and \linebreak \mbox{$FactorizeA_{11}$} are 
		local and 
		incur no additional I/O cost.
		
		{\noindent}
		\textbf{Steps 7 and 9.} Factorized  \mbox{$A_{10}$} and \mbox{$A_{01}$} 
	are scattered among all processors.
		Each processor requires  \mbox{$\frac{v(N-tv)}{c\sqrt{P1}}$} elements 
		from  
		\mbox{$A_{10}$} and \mbox{$A_{10}$}.
		\emph{I/O cost per proc.:} 
		\mbox{$\frac{2(N-tv)Nv}{P\sqrt{M}}$}. 
		
		{\noindent}
		Summing steps 1-11:
		\mbox{$Q_{step}(t) = \frac{2Nv(N-tv)}{P\sqrt{M}} + 
			\mathcal{O}\left(\frac{Nv}{P} \right)$}.
\end{proof}

\begin{figure*}[t]
	\centering	
	\subfloat[Communication volume per node for varying node 
	counts $P$ and a 
	fixed  $N=16,384$. Only the leading factors of the models 
	are shown. The models are scaled by the element size (8 
	bytes).]
	{\includegraphics[width=\columnwidth]{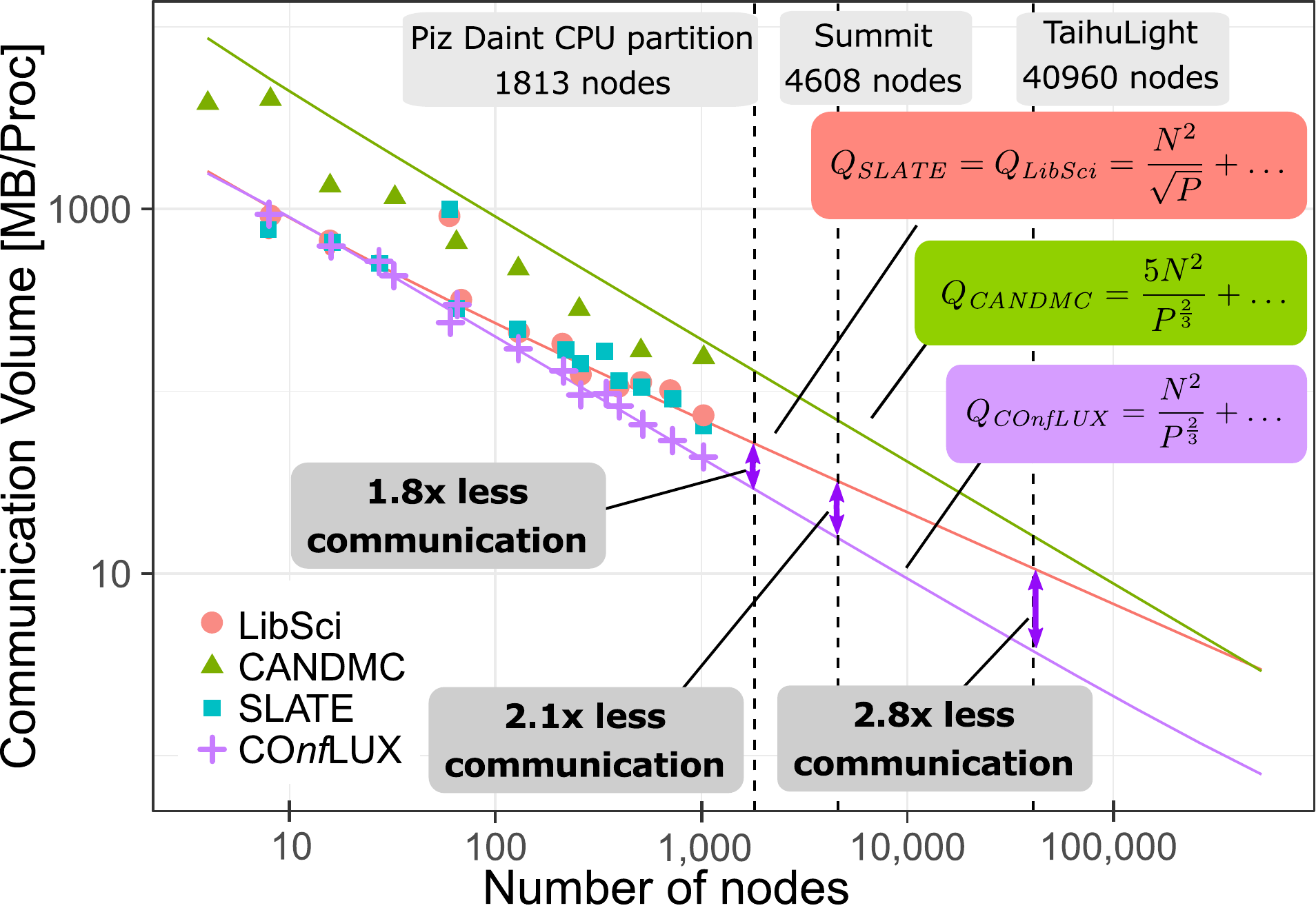}
		\label{fig:square_strong}}
	\hspace{1em}	
	\subfloat[Communication volume per node for weak scaling 
	(constant work per 
	node), 
	{\mbox{$N=3200 \cdot \sqrt[3]{P}$}}. 2.5D algorithms 
	(CANDMC and \conflux) retain constant 
	communication volume per processor.]
	{\includegraphics[width=\columnwidth ]{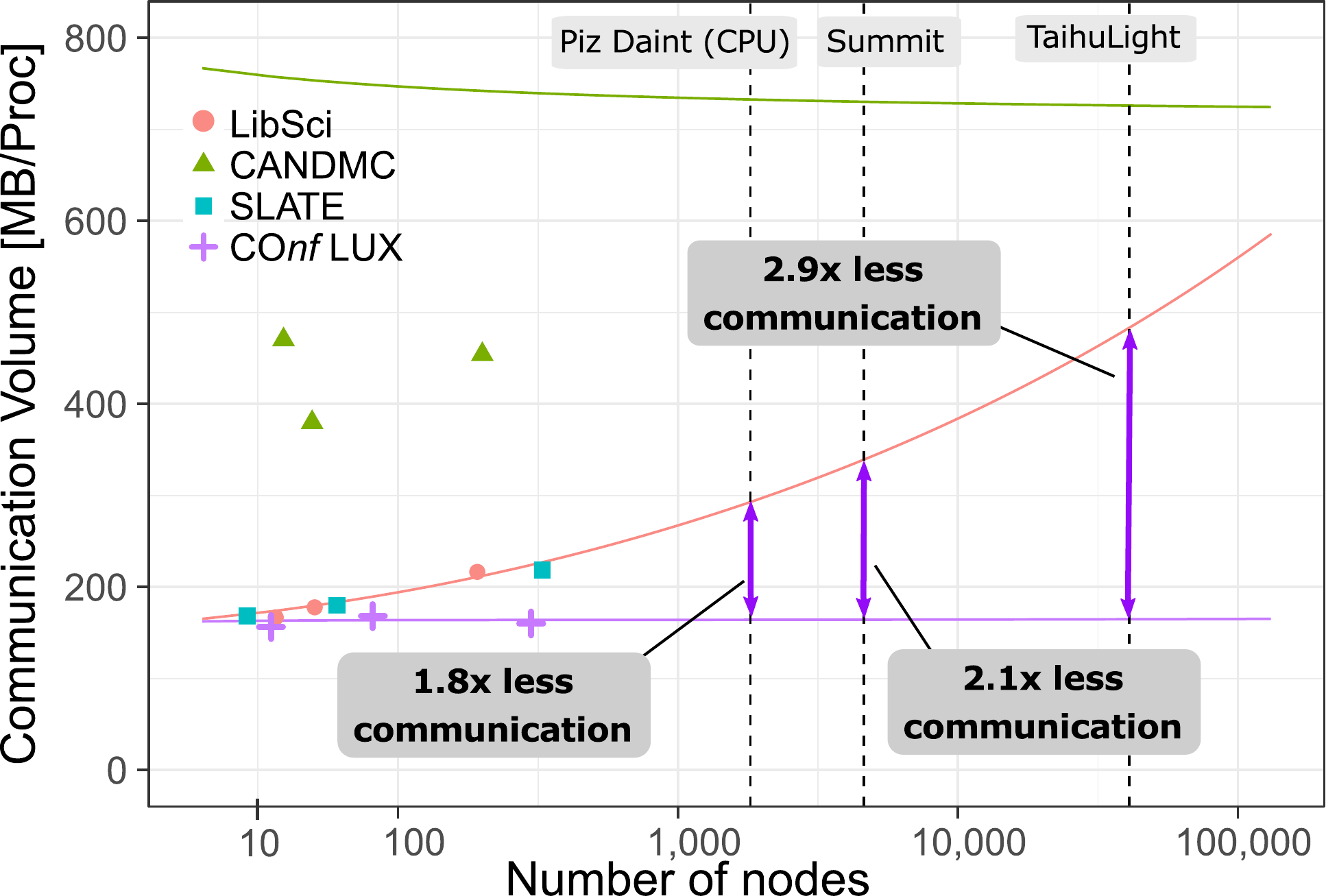}	
		\label{fig:square_weakp1}}
	\caption{		
		\textmd{Communication volume measurements across 
			different scenarios 
			for LibSci, SLATE, CANDMC, and \conflux. In all 
			considered 
			scenarios, enough memory $M\ge 
			\frac{N^2}{P^{2/3}}$ was present 
			to allow the maximum number of replications 
			$c = P^{1/3}$.}					
	}
	\label{fig:commVolPlots}
	
\end{figure*}

\section{Experimental Evaluation}
\label{sec:evaluation}

\begin{figure}[t]
	\includegraphics[width=\columnwidth]{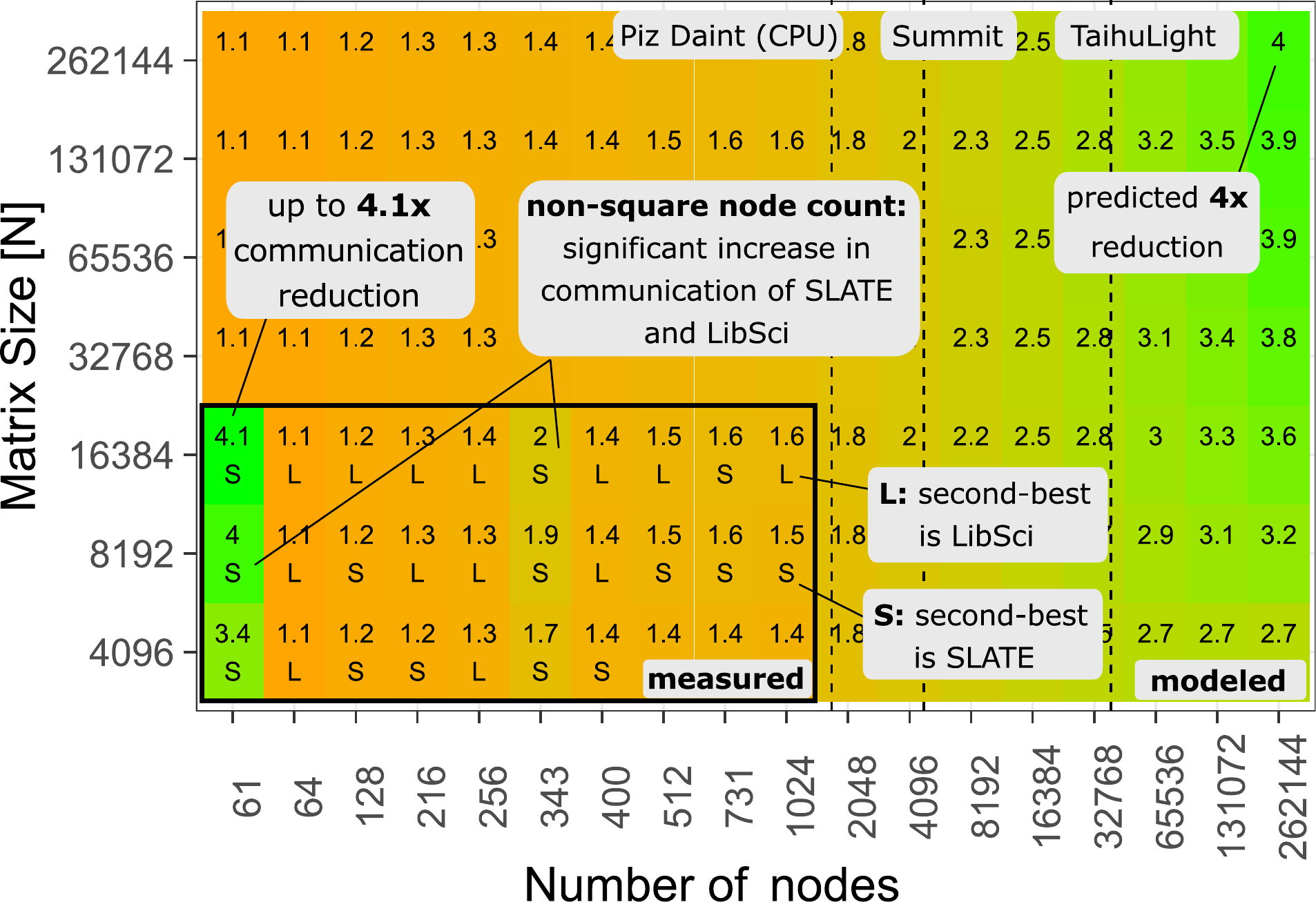}	
	
	\caption{		
		\textmd{Communication reduction vs. second-best 
			algorithm (L=LibSci, 
			S=SLATE), for varying $P$, $N$, for both measured 
			and 
			predicted 
			scenarios.}}
	\label{fig:heatmap}
\end{figure}

We implement \conflux and compare it with state-of-the-art implementations of 
distributed LU 
factorization. We measure their I/O complexity by counting their 
aggregated communication volume in distributed runs. We provide both measured 
values and theoretical cost models, on a variety of problem sizes and number of 
nodes based on scientific computing applications.

\noindent
\macb{Implementation.}
We implement \conflux in C++ using MPI one-sided~\cite{mpi3-rma-overview} for 
inter-node communication.
To secure the best performance for all combinations of processor counts and 
matrix sizes, we use Processor Grid 
Optimization~\cite{COSMA}, which finds the 3D 
processor grid with the lowest communication cost by possibly disabling a minor 
fraction of nodes. Other implementations, which greedily try to utilize all 
resources, often find communication-suboptimal decompositions for 
difficult-to-factorize number 
of ranks.  

\noindent
\macb{Infrastructure and Measurement.} 
We run our experiments on the CSCS Piz Daint supercomputer, which 
comprises 5,704 XC50 nodes equipped with Intel Xeon E5-2690 v3 processors (12 
cores, 64 GiB DDR3 RAM), interconnected by the Cray Aries network with a 
Dragonfly network topology. To measure communication volume, we instrument the 
implementations with the Score-P library~\cite{score-p} and count 
the aggregate 
bytes sent over the network.

\noindent
\macb{Comparison Targets.} For comparison, we use 1) the vendor-optimized 
ScaLAPACK implementation on Piz Daint (Cray \texttt{LibSci} v19.06.1). While 
the library is proprietary, our measurements 
reaffirm that, like ScaLAPACK, the implementation uses the suboptimal 2D 
processor decomposition; 2) SLATE~\cite{slate} --- a 
state-of-the-art distributed linear algebra framework 
targeted at exascale supercomputers; 3) the latest 
version of the CANDMC 
library~\cite{candmccode}, which uses the asymptotically-optimal 2.5D
decomposition. The implementations and their characteristics are listed in 
Table 
\ref{tab:results}.

\noindent
\macb{Problem Sizes.} 
We choose our benchmarks to reflect problems in scientific computing. 
Specifically, we choose $4,096\le N \le 16,384$. For example, Physical 
Chemistry or Density 
Functional Theory (DFT) simulations require factorizing matrices of atom 
interactions, yielding sizes of $N\ge 10,000$~\cite{gb19}. 
For node count, we measure the algorithms starting from small square and cube 
nodes ($P=4,8$) up to $P=1,024$, reflecting different scales for various 
use-cases.
In other domains, matrix sizes can be larger --- the High-Performance Linpack 
benchmark uses a maximal size of $N=16,473,600$~\cite{top500}, and in quantum 
physics matrix size scales with $2^{\text{qubits}}$. Therefore, we 
extrapolate our models to match these problem sizes and the number of 
processors on the current top supercomputers (Summit, TaihuLight) and show 
predicted communication results.

\noindent \macb{Theoretical Models.}
Together with empirical measurements, we put significant effort into 
understanding the underlying communication patterns of the compared LU 
factorization implementations. {Both LibSci and SLATE base on the standard 
	partial pivoting algorithm using the 2D decomposition{~\cite{scalapack}}. 
	For 
	CANDMC, we 
	use the model provided by the authors{~\cite{2.5DLU}}. For 
	{\conflux}, we use the results from Section{~\ref{sec:conflux}}. These 
	models 
	are summarized in Table{~\ref{tab:results}}. }

\section{Results}
\label{sec:results}

Our experiments confirm a clear advantage of \conflux in terms of 
communication volume over all other implementations tested. Not only do the 
measured values exhibit a significant communication reduction (1.42 times 
compared with the second-best implementation for $P=$ 1,024), 
but the performance models predict even greater benefits for larger runs
 (expected 
2.1 times communication reduction for a full-machine run on the Summit 
supercomputer).

\noindent\macb{Scaling Experiments.}
Fig.~\ref{fig:square_strong} presents the {measured }communication volume 
per node, as well as our derived cost models (Table{~\ref{tab:results}}) 
	presented with solid lines, 
for $N=$ 16,384. Observe that \conflux communicates the least for all values of 
$P$. Furthermore, thanks to the Processor Grid Optimization, it always finds 
the best processor grid given available resources, resulting in 
smooth and predictable performance. Other implementations try to aggressively 
use all available resources, which leads to suboptimal performance and 
visible outliers with highly increased communication, as seen in the inset. 
Note that 
since 
both LibSci and SLATE use similar 2D decomposition, their communication 
volumes are mostly equal, with a slight advantage of SLATE 
for non-square processor grids.
In 
Fig.~\ref{fig:square_weakp1}, we show the weak scaling characteristics of the 
analyzed implementations. Observe that for a fixed work per node, the 2D 
algorithms - LibSci and SLATE - scale sub-optimally. 

\noindent\macb{Implications for Exascale.}
Figure~\ref{fig:heatmap} summarizes the communication volume reduction of 
\conflux compared with the second-best implementation, both for measurements 
and 
theoretical predictions. It can be seen that in all combinations of $P$ and 
$N$, \conflux always communicates less. For all measured data points, 
the asymptotically optimal CANDMC performed worse than LibSci 
or SLATE. 
The figure also presents the predicted communication cost of all considered 
implementations for up to $P=$ 262,144, based on our theoretical models.
Considering the use of one (MPI) process per socket and/or accelerator of each 
node,
such scales will be attainable in the near future.
Observe that (a) the asymptotically optimal CANDMC is predicted to 
communicate less than suboptimal 2D implementations only for $P>$ 450,000 ranks 
for $N=16,384$, \emph{showing that asymptotic optimality is not 
	enough to secure practical performance}; and (b) for a full-scale run on 
Summit, 
\conflux is expected to communicate 2.1 times less than SLATE, a
library designed specifically for such machines.

\section{Related Work}

\begin{table*}	
	\footnotesize
	\begin{tabular}
		{p{1.1cm}p{5cm}p{4.9cm}p{5.25cm}}
		\toprule
		& \textbf{Pebbling}~\cite{sethi1975complete, 
			bruno1976code, 
			jia1981complexity, redbluewhite, COSMA} & 
		\textbf{Projection-based}~\cite{demmel1, demmel2, 
			demmel3, demmel4, 
			ballard2011minimizing, olivry2020automated} 
		& 
		\textbf{Problem specific}~\cite{aggarwal1988input, 
			benabderrahmane2010polyhedral, 
			mehta2014revisiting, darte1999complexity, gb19}
		\\
		\midrule
		\textbf{Scope} & 
		\faThumbsOUp \faThumbsOUp ~General cDAGs & 
		\makecell[tl] {\faThumbsOUp 
			~Programs with 
			static geometric\\\hspace{1.4em}structure of 
			iteration space}  & 
		\faThumbsDown ~Individually 
		tailored for given problem  \\
		\addlinespace
		\textbf{Key\newline Features} &
		\makecell[tl] {
			\faThumbsOUp ~General scope - can handle 
			irregular \\
			\hspace{1.4em}program structures \\
			\faThumbsOUp ~Expresses complex data
			dependencies \\
			\faThumbsOUp ~Directly exposes schedules \\
			\faThumbsOUp ~Intuitive \\
			\faThumbsDown ~P-SPACE complete 
			in general case \\
			\faThumbsDown ~No guarantees that a solution  
			exists \\
			\faThumbsDown ~No well-established method
			how to \\
			\hspace{1.4em}automatically translate
			code to cDAGs 
		} 
		& 
		\makecell[tl] {
			\faThumbsOUp ~Well-developed theory and tools  \\
			\faThumbsOUp ~Guaranteed to find solution \\
			\hspace{1.4em}for given class of 
			programs \\
			\faThumbsDown ~Bounds are often not tight \\
			\faThumbsDown ~Fails to capture dependencies\\
			\hspace{1.4em}between statements \\
			\faThumbsDown ~Limited scope
		}
		& 
		\makecell[tl] {
			\faThumbsOUp ~Takes advantage of
			problem-specific  \\
			\hspace{1.4em}features\\
			\faThumbsOUp ~Tends to provide best 
			practical results \\
			\faThumbsDown ~Requires large manual effort \\
			\hspace{1.4em}for each algorithm separately \\
			\faThumbsDown ~Difficult to generalize \\
			\faThumbsDown ~Often based on heuristics\\
			\hspace{1.4em}with no guarantees on optimality 
		}
		\\
		\bottomrule
	\end{tabular}
	\caption{Overview of different approaches to modeling 
		data 
		movement.
	}
	
	\label{tab:stateoftheart}
\end{table*}

Data movement analysis, while being 
prevalent for decades, has branched in multiple directions, 
In summary, previous 
work can be categorized into three classes (see Table~\ref{tab:stateoftheart}): 
(1) 
work based on direct 
pebbling or variants of it, such as Vitter's block-based 
model~\cite{vitter1998external}; (2) works using geometric arguments 
of projections based on the Loomis Whitney 
inequality~\cite{loomisWhitney}; and 
(3) works applying optimizations limited to specific structural 
properties of computations such as affine loops~\cite{affineloops}, 
and more generally, the polyhedral model program 
representation~\cite{benabderrahmane2010polyhedral, 
	mehta2014revisiting, olivry2020automated}. Although the scopes of those 
	approaches 
significantly overlap --- for example, kernels like matrix multiplication can 
be captured by most of the models --- there are still important 
differences both in methodology and end-results they provide, as summarized in 
Table~\ref{tab:stateoftheart}.

Dense linear algebra operators are among the standard core kernels in 
scientific applications.
Ballard et al.{~\cite{ballard2011minimizing}} 
present a comprehensive overview of their asymptotic I/O lower bounds and 
I/O minimizing schedules, both for sparse and dense matrices. Recently, Olivry 
et al. introduced IOLB~\cite{olivry2020automated} --- an automated framework 
for assessing sequential lower bounds for polyhedral programs. However, their 
computational model disallows recomputation, and therefore cannot capture 
programs like the one presented in Section~\ref{sec:output_reuse}. 

As such, linear solvers are implemented in various libraries for shared-memory 
environments~\cite{plasma,lapack,anotherLU, eigen, mkl, cusolver,magma}. For 
distributed 
memory, vendor-optimized libraries~\cite{libsci,mkl} typically implement the 
ScaLAPACK interface~\cite{scalapack}, and are based on 2D decomposition, as we 
empirically verify (Section \ref{sec:evaluation}).  On the algorithmic side, 
research is conducted into implementing communication-avoiding solvers with 
2.5D~\cite{candmc,2.5DLU}, and 3D 
decomposition~\cite{ballard18qr3d,choleskyQRnew} strategies. For heterogeneous 
hardware (e.g., GPU-accelerated) systems, recent frameworks focus on 
implementing modified interfaces for asynchronous offloading~\cite{dplasma}, 
and fine-grained task parallelism~\cite{chameleon,slate}.


\section{Conclusions}

{In this work, we present a novel method of analyzing DAAP --- a general 
	class of programs that covers many fundamental computational motifs.}
We show, both theoretically and in practice, that our pebbling-based approach 
for deriving the I/O lower bounds is \textbf{more general:} programs with 
disjoint array accesses cover a wide variety of applications, 
\textbf{more powerful:} it can explicitly capture inter-statement dependencies, 
\textbf{more precise:} it 	derives tighter I/O bounds, and \textbf{more 
	constructive:} \xpart provides powerful hints for obtaining parallel 
	schedules. 

When applying the approach to LU factorization, we were able to derive new 
lower bounds, as well as the \conflux schedule. Not only is \conflux 
asymptotically optimal, but we also see that in practice, the reduction in the 
leading term yields communication volumes that are better than state-of-the-art 
2D \textit{and} 3D decomposition, by a factor of up to 4.1$\times$. This 
promising result mandates the exploration of the parallel pebbling strategy to 
algorithms such as Cholesky factorization, other nontrivial dense linear 
algebra 
kernels, and beyond.

\bibliographystyle{IEEEtrans}

\bibliography{refs}

\end{document}